\documentclass{article}
\usepackage[utf8]{inputenc}

\usepackage{amssymb}
\usepackage{amsthm}
\usepackage{amsmath}
\usepackage{thm-restate}
\usepackage{booktabs}
\usepackage{complexity}
\usepackage{multirow}
\usepackage{mathtools}

\newcommand{\ALG}{\textsc{ALG}}
\newcommand{\OPT}{\textsc{OPT}}

\DeclareMathOperator{\SSW}{SW}
\DeclareMathOperator{\MMW}{MW}
\DeclareMathOperator{\Parallelop}{Parallel}
\DeclareMathOperator{\haters}{haters}
\DeclareMathOperator{\indiff}{indiff}
\DeclareMathOperator{\single}{single}

\DeclareMathOperator*{\argmin}{arg\,min}
\newcommand{\Location}{\mathbf{x}}
\newcommand{\Aversion}{\mathbf{a}}
\newcommand{\LieAversion}{\mathbf{a}'}
\newcommand{\Parallel}[1]{\Parallelop(#1)}

\newcommand{\etal}{et al.}
\newcommand{\Instance}{(n,k,\Location, \Aversion)}
\newcommand{\LieInstance}{(n,k,\Location,\LieAversion)}
\newcommand{\dist}{\Delta}

\theoremstyle{definition}
\newtheorem{definition}{Definition}
\newtheorem{theorem}{Theorem}

\newtheorem{lemma}{Lemma}
\newtheorem{fact}{Fact}
\newtheorem{innercustomthm}{Mechanism}{\bfseries}{\itshape}
\newenvironment{custommech}[1]
  {\renewcommand\theinnercustomthm{#1}\innercustomthm}
  {\endinnercustomthm}
\newtheorem{mechanism}{Mechanism}{\bfseries}{\itshape}

\begin{document}

\title{The Obnoxious Facility Location Game \\with Dichotomous Preferences\footnote{Department of Computer Science, University of Texas at Austin.  Email: \{fuli, plaxton, vaibhavsinha\}@utexas.edu. This is an extended version of a paper presented at the 22nd Italian Conference on Theoretical Computer Science in September 2021.
}}

\author{Fu Li \and
C. Gregory Plaxton \and
Vaibhav B. Sinha}

\date{May 2023}

\maketitle

\begin{abstract}
We consider a facility location game in which $n$ agents reside at known locations on a path, and $k$ heterogeneous facilities are to be constructed on the path. Each agent is adversely affected by some subset of the facilities, and is unaffected by the others. We design two classes of mechanisms for choosing the facility locations given the reported agent preferences: utilitarian mechanisms that strive to maximize social welfare (i.e., to be efficient), and egalitarian mechanisms that strive to maximize the minimum welfare. For the utilitarian objective, we present a weakly group-strategyproof efficient mechanism for up to three facilities, we give strongly group-strategyproof mechanisms that achieve approximation ratios of $5/3$ and $2$ for $k=1$ and $k > 1$, respectively, and we prove that no strongly group-strategyproof mechanism achieves an approximation ratio less than $5/3$ for the case of a single facility. For the egalitarian objective, we present a strategyproof egalitarian mechanism for arbitrary $k$, and we prove that no weakly group-strategyproof mechanism achieves a $o(\sqrt{n})$ approximation ratio for two facilities. We extend our egalitarian results to the case where the agents are located on a cycle, and we extend our first egalitarian result to the case where the agents are located in the unit square.
\end{abstract}

\section{Introduction}

The facility location game ($\FLG$) was introduced by Procaccia and
Tannenholtz~\cite{PROCACCIA2009}. In this setting, a central
planner wants to build a facility that serves agents located on a path. The agents report their locations, which are fed to a
mechanism that decides where the facility should be built. Procaccia
and Tannenholtz studied two different objectives that the planner
seeks to minimize: the sum of the distances from the facility to all
agents, and the maximum distance of any agent to the facility.

Every agent aims to maximize their welfare, which increases as their
distance to the facility decreases. An agent or a coalition of agents
can misreport their location(s) to try to increase their welfare.
The strategyproof~(SP) property says that no agent can increase
their welfare by lying about their dislikes. The weakly
group-strategyproof (WGSP) property says that if a non-empty
coalition of agents lies, then at least one agent in the coalition
does not increase their welfare.  The strongly group-strategyproof
(SGSP) property says that if a coalition of agents lies and some
agent in the coalition increases their welfare, then some agent in
the coalition decreases their welfare.  It is natural to seek SP,
WGSP, or SGSP mechanisms, which incentivize truthful reporting.
Often such mechanisms cannot simultaneously optimize the planner's
objective.  In these cases, it is desirable to approximately optimize
the planner's objective.

In real scenarios, an agent might dislike a certain facility, such as
a power plant, and want to stay away from it. This variant, called the
obnoxious facility location game ($\OFLG$), was introduced by Cheng et
al., who studied the problem of building an obnoxious facility on a
path~\cite{CHENG2011}. In the present paper, we consider the
problem of building multiple obnoxious facilities on a path. With
multiple facilities, there are different ways to define the welfare
function. For example, in the case of two facilities, the welfare of
the agent might be the sum, minimum, or maximum of the distances to the
two facilities. In our work, as all the facilities are obnoxious, a
natural choice for welfare is the minimum distance to
any obnoxious facility: the closest facility to an agent causes them
the most annoyance, and if it is far away, then the agent
is satisfied.

A facility might not be universally obnoxious. Consider, for example,
a school or sports stadium. An agent with no children might consider a
school to be obnoxious due to the associated noise and traffic, while
an agent with children might not consider it to be obnoxious. Another
agent who is not interested in sports might similarly consider a
stadium to be obnoxious. We assume that each agent has dichotomous
preferences; they dislike some subset of the facilities and are
indifferent to the others. Each agent reports a subset of
facilities to the planner. As the dislikes are private information,
the reported subset might not be the subset of facilities that the
agent truly dislikes. On the other hand, we assume that the agent
locations are public and cannot be misreported.

In this paper, we study a variant of $\FLG$, which we call $\OOFLG$
(Dichotomous Obnoxious Facility Location Game), that
combines the three aspects mentioned above: multiple (heterogeneous)
obnoxious facilities, minimum distance as welfare, and dichotomous
preferences.
We seek to design mechanisms that perform well with respect to either a utilitarian or egalitarian objective.
The utilitarian objective
is to maximize the social welfare, that is, the total welfare of
all agents.  A mechanism that maximizes social welfare is said to be
efficient.  The egalitarian objective is to maximize the minimum 
welfare of any agent. For both objectives, we
seek mechanisms that are SP, or better yet, weakly or
strongly group-strategyproof (WGSP / SGSP).

\subsection{Our contributions}

We study $\OOFLG$ with $n$ agents. 
In Section~\ref{sec:sum-min}, we consider the utilitarian objective. 
We present $2$-approximate SGSP mechanisms for any number of facilities when the agents are located on a path, cycle, or square.
We obtain the following two additional results for the  path setting.
We obtain a mechanism that is WGSP for any number of facilities and
efficient for up to three facilities.  To show that this mechanism is
WGSP, we relate it to a weighted approval voting mechanism.  To prove
its efficiency, we identify two crucial properties that the welfare
function satisfies, and we use an exchange argument.  For the path
setting, we show that no SGSP mechanism can achieve an approximation
ratio better than $5/3$, even in the single-facility case, and we
  present a single-facility SGSP mechanism that achieves an
  approximation ratio of $5/3$, matching the lower bound. The argument
  underlying our $5/3$ lower bound demonstrates that any
  single-facility SGSP mechanism needs to essentially disregard the
  agent preferences; in other words, the location of the facility
  has to be (essentially) determined by the agent locations.

The single-facility mechanism that we use to establish the matching
$5/3$ upper bound disregards the agent preferences entirely, and hence
is SGSP.  Our proof of the $5/3$ upper bound is by far the most
technical argument in the paper. Given the agent locations, we first
use a sequence of lemmas to characterize the best possible
approximation ratio that can be guaranteed (for all possible choices of the agent preferences) if the mechanism locates the facility at the left
endpoint, right endpoint, or center of the path. (We also give a fast
algorithm for computing these three approximation ratios, which allows
for a fast implementation of our mechanism.)  We exploit this
characterization to show that it is sufficient to bound the
approximation ratio achieved by the mechanism on instances where all
of the agents to the left (resp., right) of the center are located at no
more than two distinct locations.  We then show that it is sufficient
to further restrict our attention to ``balanced'' instances where the
average agent location (i.e., the center of gravity of the agents) is
at the center. Under these restrictions, we are able to show that if
the mechanism cannot guarantee a $5/3$ approximation ratio by building
the facility at the left or right endpoint, then it can guarantee a $5/3$ approximation ratio by building at the center.

In Section~\ref{sec:min-min}, we consider the egalitarian objective. We provide optimal SP mechanisms for
any number of facilities when the agents are located on a path, cycle,
or square. 
We prove that the approximation ratio achieved by any WGSP mechanism is $\Omega(\sqrt{n})$, even for two facilities.  
Also, we
present a straightforward $O(n)$-approximate WGSP 
mechanism. Both of the results for WGSP mechanisms hold
for $\OOFLG$ when the agents are located on a path or
cycle. Table~\ref{tab:results} summarizes our results.

\begin{table}
\centering
\caption{Summary of our results for $\OOFLG$ when the agents are
  located on a path. The heading LB (resp., UB) stands for lower
  (resp., upper) bound. The results in the egalitarian column also
  hold when the agents are located on a cycle. Boldface results hold
  when the agents are located on a path, cycle, or square. The
    tight $5/3$ upper bound for the SGSP utilitarian case holds when
    there is a single facility and the agents are located on a path,
    while the upper bound of $2$ holds for an arbitrary number of
    facilities when the agents are located on a path, cycle, or
    square.}
\vspace{1pt}
\begin{tabular}{ l | c c | c c |}
\cline{2-5}
    & \multicolumn{2}{c|}{Utilitarian} & \multicolumn{2}{c|}{Egalitarian} \\
    \cline{2-5}
    & LB & UB & LB & UB \\
\hline
\multicolumn{1}{|l|}{SP}  & \multirow{2}{*}{$\mathbf{1}$} & \multirow{2}{*}{$1$ for $k \le 3$} & $\mathbf{1}$ & $\mathbf{1}$\\ 
\cline{1-1} \cline{4-5}
\multicolumn{1}{|l|}{WGSP} & & & \multirow{3}{*}{$\Omega(\sqrt{n})$} & \multirow{3}{*}{$O(n)$}\\ \cline{1-3} 
\multicolumn{1}{|l|}{\multirow{2}{*}{SGSP}} & \multirow{2}{*}{$5/3$} & $5/3$ for $k = 1$ & & \\
\multicolumn{1}{|l|}{} & & $\mathbf{2}$ & & \\
\hline
\end{tabular}
\label{tab:results}
\end{table}

\subsection{Related work}
\label{sec:related-works}

Procaccia and Tannenholtz~\cite{PROCACCIA2009}
introduced $\FLG$. Many generalizations and extensions of
$\FLG$ have been studied \cite{ALON2010,DOKOW2012,FELDMAN2013,FILOS0217,FOTAKIS2010,FOTAKIS2014,LU2010,ZHANG2014}; here
we highlight some of the most relevant work.
Cheng~\etal\ introduced $\OFLG$ and presented a WGSP
mechanism to build a single facility on a path~\cite{CHENG2011}. Later they extended the model to cycles and
trees~\cite{CHENG2013}. A complete characterization of
single-facility SP/WGSP mechanisms for paths has
been developed~\cite{HAN2012,IBARA2012}. Duan~\etal\ studied the
problem of locating two obnoxious facilities at least distance $d$
apart~\cite{DUAN2019}. Other variants of $\OFLG$ have been
considered \cite{CHENG2013BOUNDEDSERVICE,FUKUI2019,OOMINE2016,YE2015}.

Agent preferences over the facilities were introduced to $\FLG$ in~\cite{FEIGENBAUM2015} and~\cite{ZOU2015}.  Serafino and Ventre studied
$\FLG$ for building two facilities where each agent likes a subset of
the facilities~\cite{SERAFINO201627}. Anastasiadis and Deligkas extended this model to allow the agents to like, dislike, or be
indifferent to the facilities~\cite{ANASTASIADIS2018}. The
aforementioned works address linear (sum) welfare functions. Yuan~\etal\ studied non-linear welfare functions (maximum and minimum) for building two
non-obnoxious facilities~\cite{YUAN2016}; their results have
subsequently been strengthened~\cite{CHEN2020,LI2020}.  In the present
paper, we initiate the study of a non-linear welfare function (minimum) for building
multiple obnoxious facilities.

\newcommand{\Soln}{\mathbf{y}}

\section{Preliminaries}
\label{sec:preliminaries}

The problems considered in this paper involve a set of agents located
on a path, cycle, or square.  In the path (resp., cycle, square)
setting, we assume without loss of generality that the path (resp.,
cycle, square) is the unit interval (resp., unit-circumference circle,
unit square).  We map the points on the unit-circumference circle to
$[0,1)$, in the natural manner.  Thus, in the path (resp., cycle,
square) setting, each agent $i$ is located in $[0,1]$ (resp., $[0,1)$,
$[0,1]^2$).  The distance between any two points $x$ and $y$ is
denoted $\dist(x,y)$.  In the path and square settings, $\dist(x,y)$ is
defined as the Euclidean distance between $x$ and $y$.  In the cycle
setting, $\dist(x,y)$, is defined as the length of the shorter arc between
$x$ and $y$.  In all settings, we index the agents from $1$. Each
agent has a specific location in the path, cycle, or square. A
\emph{location profile} $\Location$ is a vector $(x_1,\ldots,x_n)$ of
points, where $n$ denotes the number of agents and $x_i$ is the
location of agent $i$. 
Sections~\ref{subsec:sum-min-interval}
and~\ref{subsec:min-min-interval} (resp., Sections~\ref{subsec:sum-min-cycle}
and~\ref{subsec:min-min-cycle},  Sections~\ref{subsec:sum-min-square}
and~\ref{subsec:min-min-square}) present our results for the path (resp., cycle, square)
setting.

Consider a set of agents $1$ through $n$ and a set of facilities
$\mathcal{F}$, where we assume that each agent dislikes (equally)
certain facilities in $\mathcal{F}$ and is indifferent to the rest. In
this context, we define an \emph{aversion profile} $\Aversion$ as a
vector $(a_1,\ldots,a_n)$ where each component $a_i$ is a subset of
$\mathcal{F}$.  We say that such an aversion profile is \emph{true} if
each component $a_i$ is equal to the subset of $\mathcal{F}$ disliked
by agent $i$.  In this paper, we also consider \emph{reported}
aversion profiles where each component $a_i$ is equal to the set of
facilities that agent $i$ claims to dislike. Since agents can lie, a
reported aversion profile need not be true. For any aversion profile $\Aversion$ and any subset $C$ of agents $[n]$, $\Aversion_C$ (resp., $\Aversion_{-C}$) denotes the aversion profile for the agents in (resp., not in) $C$.
For a singleton set
of agents $\{i\}$, we abbreviate $\Aversion_{-\{i\}}$ as $\Aversion_{-i}$.

An instance of the dichotomous obnoxious facility location ($\DOFL$)
problem is given by a tuple $(n,k,\Location,\Aversion)$ where $n$ denotes the
number of agents, there is a set of $k$ facilities
$\mathcal{F}=\{F_1,\ldots,F_k\}$ to be built, $\Location=(x_1,\ldots,x_n)$
is a location profile for the agents, and $\Aversion=(a_1,\ldots,a_n)$ is
an aversion profile (true or reported) for the agents with respect to
$\mathcal{F}$.  A solution to such a $\DOFL$ instance is a vector
$\mathbf{y}=(y_1,\ldots,y_k)$ where component $y_j$ specifies the point at
which to build $F_j$. We say that a $\DOFL$ instance is true (resp.,
reported) if the associated aversion profile is true (resp.,
reported).
For any $\DOFL$ instance $I = \Instance$ and any $j$ in $[k]$, we define $\haters(I, j)$ as $\{i \in [n] \mid F_j \in a_i\}$,
and $\indiff(I)$ as $\{ i \in [n] \mid a_i = \emptyset \}$.

For any $\DOFL$ instance $I=(n,k,\Location,\Aversion)$ and any associated
solution $\Soln$, we define the \emph{welfare} of agent $i$,
denoted $w(I,i,\mathbf{y})$, as the minimum
distance from $x_i$ to any facility in $a_i$, that is, $\min_{j:F_j \in a_i} \dist(x_i, y_j)$.
Remark: If $a_i$ is
empty, we define $w(I,i,\mathbf{y})$ as $1/2$ in the cycle setting,
$\max(\dist(x_i, 0), \dist(x_i, 1))$ in the path setting, and the maximum
distance from $x_i$ to a corner in the square setting.

The foregoing definition of agent welfare is suitable for true $\DOFL$
instances, and is only meaningful for reported $\DOFL$ instances where
the associated aversion profile is close to true.  In this paper,
reported aversion profiles arise in the context of mechanisms that
incentivize truthful reporting, so it is reasonable to expect such
aversion profiles to be close to true.  We define the \emph{social
 welfare} (resp., \emph{minimum welfare}) as the sum (resp., minimum)
of the individual agent welfares. When the facilities are built at
$\Soln$, the social welfare and minimum welfare are denoted by
$\SSW(I,\Soln)$ and $\MMW(I,\Soln)$, respectively. Thus
$\SSW(I,\Soln)=\sum_{i \in [n]}w(I,i,\Soln)$ and
$\MMW(I,\Soln)=\min_{i \in [n]}w(I,i,\Soln)$.

\begin{definition}
For $\alpha \ge 1$, a $\DOFL$ algorithm $A$ is $\alpha$-efficient if for any $\DOFL$ instance $I$,
\begin{equation*}
    \max_{\mathbf{y}} \SSW(I, \mathbf{y}) \le \alpha \SSW (I, A(I)).
\end{equation*} Similarly, $A$ is $\alpha$-egalitarian if for any $\DOFL$ instance $I$,
\begin{equation*}
    \max_{\mathbf{y}} \MMW(I, \mathbf{y}) \le \alpha \MMW (I, A(I)).
\end{equation*}
\end{definition}
A $1$-efficient (resp., $1$-egalitarian) $\DOFL$ algorithm, is said to be efficient (resp., egalitarian).

We are now ready to define a $\DOFL$-related game, which we call $\OOFLG$. It is
convenient to describe a $\OOFLG$ instance in terms of a pair $(I,I')$ of
$\DOFL$ instances where $I=(n,k,\Location,\Aversion)$ is true and
$I'=(n,k,\Location,\Aversion')$ is reported.  There are $n$ agents indexed
from $1$ to $n$, and a planner. There is a set of $k$ facilities
$\mathcal{F}=\{F_1,\ldots,F_k\}$ to be built. The numbers $n$ and $k$
are publicly known, as is the location profile $\Location$ of the
agents. Each component $a_i$ of the true aversion profile $\Aversion$ is
known only to agent $i$. Each agent $i$ submits component $a_i'$ of
the reported aversion profile $\LieAversion$ to the planner.  The planner,
who does not have access to $\Aversion$, runs a $\DOFL$ algorithm, call it
$A$, to map $I'$ to a solution. The input-output behavior of $A$
defines a $\OOFLG$ mechanism, call it $M$; in the special case where
$k=1$, we say that $M$ is a single-facility $\OOFLG$ mechanism.  We would
like to choose $A$ so that $M$ enjoys strong game-theoretic
properties. We say that $M$ is $\alpha$-efficient (resp.,
$\alpha$-egalitarian, efficient, egalitarian) if $A$ is $\alpha$-efficient (resp.,
$\alpha$-egalitarian, efficient, egalitarian). As indicated earlier, such properties (which
depend on the notion of agent welfare) are only meaningful if the
reported aversion profile is close to true.  To encourage truthful
reporting, we require our mechanisms to be SP, as defined below; we
also consider the stronger properties WGSP and SGSP.

The SP property says that no agent can increase their welfare by lying
about their dislikes, regardless of the fixed aversion profile
  reported by the remaining agents.

\begin{definition}
A $\OOFLG$ mechanism $M$ is SP if for any $\OOFLG$ instance $(I, I')$ with $I = \Instance$, and $I' = \LieInstance$, 
and any agent $i$ in $[n]$ such that $\LieAversion = (\Aversion_{-i}, a_i')$, we have
\begin{equation*}
    w(I, i, M(I)) \ge w(I, i, M(I')).
\end{equation*}
\end{definition}

The WGSP property says that if a non-empty coalition $C \subseteq [n]$ of agents
lies, then 
at least one agent in $C$ does not increase their welfare.

\begin{definition}
A $\OOFLG$ mechanism $M$ is WGSP if for any $\OOFLG$ instance $(I, I')$ with $I = \Instance$, and $I' = \LieInstance$, and any non-empty coalition $C \subseteq [n]$ such that $\LieAversion = (\Aversion_{-C}, \LieAversion_C)$, there exists an agent $i$ in $C$ such that
\begin{equation*}
    w(I, i, M(I)) \ge w(I, i, M(I')).
\end{equation*}
\end{definition}

The SGSP property says that if a coalition $C \subseteq [n]$ of agents lies 
and some agent in $C$ increases their welfare then some agent in $C$ decreases their welfare.

\begin{definition}
A $\OOFLG$ mechanism $M$ is SGSP if for any $\OOFLG$ instance $(I, I')$ with $I = \Instance$, and $I' = \LieInstance$, and any coalition $C \subseteq [n]$ such that $\LieAversion = (\Aversion_{-C}, \LieAversion_C)$, if there exists an agent $i$ in $C$ such that
\begin{equation*}
    w(I, i, M(I)) < w(I, i, M(I')),
\end{equation*} then there exists an agent $i'$ in $C$ such that
\begin{equation*}
    w(I, i', M(I)) > w(I, i', M(I')).
\end{equation*}
\end{definition}

Every SGSP mechanism is WGSP and every WGSP mechanism is SP.

\section{Weighted Approval Voting}
\label{sec:weight-app-vote}

Before studying efficient mechanisms for our problem, we review a
variant of the approval voting mechanism~\cite{BRAMS1978}.  An
instance of the Dichotomous Voting (DV) problem is a tuple $(m, n,
\mathbf{C}, \mathbf{w}^+, \mathbf{w}^-)$ where $m$ voters $1, \dots,
m$ have to elect a candidate among a set of candidates $C = \{c_1,
\dots, c_n\}$.  Each voter $i$ has dichotomous preferences, that is,
voter $i$ partitions all of the candidates into two equivalence
classes: a top (most preferred) tier $C_i$ and a bottom tier
$\overline{C_i} = C \setminus C_i$.  Each voter $i$ has associated
(and publicly known) weights $w_i^+ \ge w_i^- \ge 0$.  The symbols
$\mathbf{C}$, $\mathbf{w}^+$, and $\mathbf{w}^-$ denote length-$m$
vectors with $i$th element $C_i$, $w^+_i$, and $w^-_i$, respectively.
We now present our weighted approval voting mechanism.\footnote{Our
  mechanism differs from the homonymous mechanism of Massó et al.,
  which has weights for the candidates instead of the
  voters~\cite{MASSO2008}.}

\begin{mechanism}
\label{mech:weighted-approval-voting}
Given  a $\DV$ instance $(m, n,\mathbf{C}, \mathbf{w}^+, \mathbf{w}^-)$, every voter $i$ votes by partitioning $C$ into $C'_i$ and $\overline{C'_i}$. 
Let the weight function $w$ be such that for voter $i$ and candidate $c_j$, $w(i, j) = w_i^+$ if $c_j$ is in $C'_i$ and $w(i, j) = w_i^-$ otherwise. 
For any $j$ in $[n]$, we define $A(j) = \sum_{i \in [m]} w(i, j)$ as the approval of candidate $c_j$.
The candidate $c_j$ with highest approval $A(j)$ is declared the winner. 
Ties are broken according to a fixed ordering of the candidates (e.g., in favor of lower indices).
\end{mechanism}

We note that the approval voting mechanism can be obtained from the weighted approval voting mechanism by setting weights $w_i^+$ to $1$ and $w_i^-$ to $0$ for all voters $i$.
In Section~\ref{sec:preliminaries}, we defined SP, WGSP, and SGSP in the $\OOFLG$ setting.
These definitions are easily generalized to the voting setting.
Brams and Fishburn proved that the approval voting mechanism is SP~\cite{BRAMS1978}. 
Below we prove that our weighted approval voting mechanism is WGSP (and hence also SP).

\begin{restatable}{theorem}{wgspappvote}
\label{thm:wgsp-app-vote}
Mechanism~\ref{mech:weighted-approval-voting} is WGSP.
\end{restatable}
\begin{proof}
Assume for the sake of contradiction that there is an
instance in which a coalition of voters $U$ with true preferences
$\{(C_i,\overline{C_i})\}_{i\in U}$ all benefit by misreporting their
preferences as $\{(C'_i,\overline{C'_i})\}_{i\in U}$.  For any
candidate $c_j$, let $A(j)$ denote the approval of $c_j$ when
coalition $U$ reports truthfully, and let $A'(j)$ denote the approval
of $c_j$ when coalition $U$ misreports.

Let $c_k$ be the winning candidate when coalition $U$ reports
truthfully, and let $c_\ell$ be the winning candidate when coalition
$U$ misreports. Since every voter in $U$ benefits when the coalition
misreports, we know that $c_k$ belongs to
$\bigcap_{i\in U}\overline{C_i}$ and $c_\ell$ belongs to $\bigcap_{i\in U}C_i$.

Since $c_k$ belongs to $\bigcap_{i\in U}\overline{C_i}$, we deduce that
$A'(k)=A(k)+\sum_{i\in U:c_k\in C'_i}w_i^+-w_i^-$ and hence
$A'(k)\geq A(k)$.  Similarly, since $c_\ell$ belongs to $\bigcap_{i\in U}C_i$, we
deduce that
$A'(\ell)=A(\ell)+\sum_{i\in U:c_\ell\in\overline{C'_i}}w_i^--w_i^+$
and hence $A(\ell)\geq A(\ell')$.

Since $c_k$ wins when coalition $U$ reports truthfully, one of the
following two cases applies.

Case~1: $A(k)>A(\ell)$. Since $A'(k)\geq A(k)$ and
$A(\ell)\geq A'(\ell)$, the case condition implies that
$A'(k)>A'(\ell)$. Hence $c_\ell$ does not win when coalition $U$
misreports, a contradiction.

Case~2: $A(k)=A(\ell)$ and $c_k$ has higher priority than $c_\ell$.
Since $A'(k)\geq A(k)$ and $A(\ell)\geq A(\ell')$, the case condition
implies that $A'(k)\geq A(\ell')$ and $c_k$ has higher priority than
$c_\ell$. Hence $c_\ell$ does not win when coalition $U$ misreports, a
contradiction.
\end{proof}

\begin{theorem}
Mechanism~\ref{mech:weighted-approval-voting} is not SGSP.
\end{theorem}
\begin{proof}
Let $I$ be a $\DV$ instance with $5$ voters, candidates $c_1$ and
  $c_2$, and weights $w^+_i = 1$ and $w^-_i = 0$ for all $i$ in $\{1,
  \dots, 5\}$.  Each voter in $I$ votes truthfully, and their votes
  are: $C_1 = \{c_1\}$, $C_2 = C_3 = \{c_1, c_2\}$, and $C_4 = C_5 =
  \{c_2\}$.  Thus $A(1) = 3$ and $A(2) = 4$, and
  Mechanism~\ref{mech:weighted-approval-voting} declares $c_2$ the
  winner.  Let $I'$ be the $\DV$ instance with the same voters,
  candidates, and weights as in $I$.  Voters $1, 2$, and $3$ form a
  coalition and vote $\{c_1\}$, while voters $4$ and $5$ vote
  $\{c_2\}$.  Then $A(1) = 3$ and $A(2) = 2$, and
  Mechanism~\ref{mech:weighted-approval-voting} declares $c_1$ the
  winner.  This result benefits voter $1$, without any loss to voters
  $2$ and $3$.  Thus Mechanism~\ref{mech:weighted-approval-voting} is
  not SGSP.
\end{proof}


\section{Efficient Mechanisms}
\label{sec:sum-min}

In this section, we present efficient mechanisms for $\OOFLG$.  In
Section~\ref{subsec:sum-min-interval}, we address the case where the
agents are located in the unit interval.  In
Section~\ref{subsec:sum-min-cycle} (resp., Section~\ref{subsec:sum-min-square}), we consider the case where the agents are
located on a cycle (resp., square).

\subsection{The unit interval}
\label{subsec:sum-min-interval}

We now present our efficient mechanism for $\OOFLG$.

\begin{mechanism}
\label{mech:sum-min}
For a given reported $\DOFL$ instance $I = \Instance$,
output the lexicographically least solution $\mathbf{y}$ in $\{0, 1\}^k$ that maximizes the social welfare $\SSW (I,\mathbf{y})$.
\end{mechanism}

Mechanism~\ref{mech:sum-min} runs in $O(nk2^k)$ time, and hence
  runs in polynomial time when $k$ is $O(\log n)$.

\begin{theorem}
\label{thm:sp-sum-min}
Mechanism~\ref{mech:sum-min} is WGSP.
\end{theorem}
\begin{proof}
To establish this theorem, we show that Mechanism~\ref{mech:sum-min} can be equivalently expressed in terms of the approval voting mechanism. Hence Theorem~\ref{thm:wgsp-app-vote} implies the theorem.

Let $(I,I')$ denote a $\OOFLG$ instance where $I=\Instance$ and $I'=\LieInstance$.
We view each agent $i \in [n]$ as a voter, and each $\mathbf{y}$ in $\{0,1\}^k$ as a candidate. We obtain the top-tier candidates $C_i$ of voter $i$, and their reported top-tier candidates $C'_i$, from $a_i$ and $a'_i$, respectively. Assume without loss of generality that $x_i \le 1/2$ (the other case can be handled similarly). Set $C_i = \{\mathbf{y} = (y_1, \dots, y_k) \in \{0,1\}^k \mid y_j = 1 \text{ for all } F_j \in a_i\}$ and similarly $C'_i = \{\mathbf{y} = (y_1, \dots, y_k) \in \{0,1\}^k \mid y_j = 1 \text{ for all } F_j \in a'_i\}$. Also set $w^+_i = 1-x_i$ and $w_i^- = x_i$. With this notation, it is easy to see that $A(\mathbf{y}) = \SSW (I', \mathbf{y})$, and that choosing the $\mathbf{y}$ with the highest social welfare in Mechanism~\ref{mech:sum-min} is the same as electing the candidate with the highest approval in Mechanism~\ref{mech:weighted-approval-voting}. 
\end{proof}

We show that Mechanism~\ref{mech:sum-min} is efficient for $k =
3$. First, we note a well-known result about the 1-Maxian problem. In
this problem, there are $n$ points located at $z_1, \dots, z_n$ in the
interval $[a,b]$, and the task is to choose a point in $[a,b]$ such
that the sum of the distances from that point to all of the $z_i$'s is
maximized. This result follows from the fact that the sum of
  convex functions is convex, and that a convex function on a closed
  interval is maximized at the one of the endpoints of the
  interval~\cite{BOYD2004}.

\begin{lemma}
\label{lem:1-maxian}
Let $[a,b]$ be a real interval, let $z_1, \dots, z_n$ belong to $[a,b]$, and let $f(z)$ denote $\sum_{i\in [n]} |z-z_i|$.
Then $\max_{z\in [a,b]} f(z)$ belongs to 
$\{f(a), f(b)\}$. 
\end{lemma}

Before proving the main theorem, we establish Lemma~\ref{lem:opt-sum-min}, which follows from Lemma~\ref{lem:1-maxian}.

\begin{restatable}{lemma}{optsummin}
\label{lem:opt-sum-min}
Let $I = (n, k, \Location, \Aversion)$ denote the reported $\DOFL$ instance,
let $Y$ denote the set of all $y$ in $[0,1]$ such that it is efficient to build all $k$ facilities at $y$, and assume that $Y$ is non-empty. Then $Y \cap \{0,1\}$ is non-empty.
\end{restatable}
\begin{proof}
Let $U$ denote $\indiff(I)$. When all of the facilities are built at $y$,
\begin{equation*}
    \SSW (I, (y, \dots, y)) = \sum_{i\in [n] \setminus U} |x_i - y| + \sum_{i \in U} w(I, i, y).
\end{equation*}

Since $Y$ is non-empty, $\max_y \SSW (I, (y, \dots, y)) = \max_\mathbf{y} \SSW(I, \mathbf{y})$.
Moreover, since $\sum_{i \in U} w(I, i, y)$ does not depend on $y$, Lemma~\ref{lem:1-maxian} implies that 
\begin{equation*}
    \max(\SSW (I, (0, \dots, 0)), \SSW (I, (1, \dots, 1))) = \max_y \SSW (I, (y, \dots, y)).
\end{equation*} 
Thus, if $\SSW (I, (0, \dots, 0)) \ge \SSW (I, (1, \dots, 1))$, it is efficient to build all $k$ facilities at $0$.
Otherwise, it is efficient to build all $k$ facilities at $1$.
\end{proof}

\begin{theorem}
\label{thm:opt-sum-min}
Mechanism~\ref{mech:sum-min} is efficient for $k=3$. 
\end{theorem}
\newcommand{\tuple}{ y_{1}^*,y_{3}^* }
\newcommand{\interval}{\ensuremath{[\tuple ]}}
\newcommand{\fixing}{|_{y_1 = y_1^*, y_3 =y_3^*}}
\newcommand{\fixingTwoSame}{|_{y_2 = y_2^*, y_3 =y_2^*}}
\newcommand{\fixingTwoSameWithZero}{|_{y_1=0, y_2 = y_3=y}}
\begin{proof}
Let $I = \Instance$ denote the reported $\DOFL$ instance and let $\mathbf{y}^* = (y_1^*,y_2^*, y_3^*)$ be an efficient solution for $I$ such
that $y_1^*\leq y_2^*\leq y_3^*$.

Consider fixing variables $y_1$ and $y_3$ in the social welfare
function $\SSW(I, \mathbf{y})$. That is, we have
\begin{equation*}
    \SSW(I, \mathbf{y})\fixing = \sum_{i\in [n]} w(I,i, \mathbf{y})\fixing.
\end{equation*}
For convenience, let $\SSW(y_2)$ denote $\SSW(I,\mathbf{y})\fixing $ and
let $w_i(y_2)$ denote $w(I,i, \mathbf{y})\fixing$ for each agent
$i$. 

Claim~1: For each agent $i$, the welfare function $w_i(y_2)$ with $y_2
\in \interval$ satisfies either~(1) $w_i(y_2) = |y_2-x_i|$ or~(2)
$w_i(y_1^*) = w_i(y_3^*) = \max_{y\in \interval} w_i(y)$.

Proof: Fix an agent $i$. We consider five cases.

Case~1: $F_2 \notin a_i$. Since the welfare of agent $i$ is
independent of the location of $F_2$, $w_i$ is a constant
function. Hence~(2) is satisfied.

Case~2: $a_i = \{F_2\}$. By definition, we have $w_i(y_2) =
|y_2-x_i|$. Hence~(1) is satisfied.

Case~3: $a_i= \{F_1,F_2\}$. By definition, we have $w_i(y_2) =
\min(|y_1^*-x_i|,|y_2-x_i|)$. 
Notice that $w_i(y_1^*) = |y_1^*-x_i|$.
Since $\min(|y_1^*-x_i|,|y_2-x_i|) \le |y_1^*-x_i|$ for all $y_2$ in $\interval$, we have $w_i(y_1^*) = |y_1^*-x_i| = \max_{y\in \interval} w_i(y)$.
Moreover, $w_i(y_3^*) =
\min(|y_1^*-x_i|,|y_3^*-x_i|)$. We consider two cases.

Case~3.1: $|y_1^*-x_i| > |y_3^*-x_i|$. 
Then $x_i$ belongs to $((y_1^* + y_3^*)/2, 1]$.
Hence $|y_2 - x_i| \le |y_1^* - x_i|$ for all $y_2$ in $\interval$.
Thus $w_i(y_2) = |y_2-x_i|$ for all $y_2$ in $\interval$, 
that is, $w_i(y_2)$ satisfies~(1).

Case~3.2: $|y_1^*-x_i|\le |y_3^*-x_i|$. Then $w_i(y_3^*)=
|y_1^*-x_i|=\max_{y\in \interval} w_i(y) = w_i(y_1^*)$ and hence
$w_i(y_2)$ satisfies~(2).

Case~4: $a_i= \{F_2,F_3\}$. This case is symmetric to Case~3, and can
be handled similarly.

Case~5: $a_i= \{F_1,F_2,F_3\}$. By definition, we have $w_i(y_2) =
\min(|y_1^*-x_i|,|y_2-x_i|,|y_3^*-x_i|)$.  Notice that $w_i(y_1^*) =
w_i(y_3^*) = \min(|y_1^*-x_i|,|y_3^*-x_i|)$.  Also notice that for any
$y_2$ in $\interval$, we have $w_i(y_2) =
\min(|y_1^*-x_i|,|y_2-x_i|,|y_3^*-x_i|) \le
\min(|y_1^*-x_i|,|y_3^*-x_i|) = w_i(y_1^*)$.  Hence~(2) holds.

This concludes our proof of Claim~1.

Claim~2: There is a solution that optimizes $\max_\mathbf{y}
\SSW(I, \mathbf{y})$ and builds facilities in at most two locations.

Proof:
We establish the claim by proving
that either $\SSW(I,(y_1^*,y_{1}^*, y_{3}^*)) \ge \SSW(I,\mathbf{y}^*) $
or $\SSW(I, (y_1^*,y_{3}^*, y_{3}^*)) \ge \SSW(I, \mathbf{y}^*)$.

Claim~1 implies that the set of agents $[n]$ can be partitioned into two sets
$(S,\overline{S})$ such that $w_i(y_2)$ satisfies (1) for all
$i$ in $S$, and $w_i(y_2)$ satisfies~(2) for all $i$ in
$\overline S$. Thus $\SSW(y_2) = \sum_{i\in [n]}
w_i(y_2) = \sum_{i\in S} w_i(y_2) + \sum_{i\in \overline S} w_i(y_2)$.
By Lemma~\ref{lem:1-maxian}, there is a $b$ in $\{ \tuple \}$ such
that $\sum_{i\in S} w_i(b)\ge \sum_{i\in S} w_i(y_2)$ for all $y_2$ in
\interval. For any $i$ in $\overline{S}$, we deduce from (2)
that $w_i(b)\ge w_i(y_2)$ for all $y_2$ in \interval. Therefore,
$\SSW(b)\ge \SSW(y_2)$ for all $y_2$ in $\interval$. This
completes our proof of Claim~2.

Having established Claim~2, we can assume without loss of generality
that $y_2^*=y_3^*$.  A simpler version of the arguments given in
Claims~1 and~2 above can be used to prove that either $(0,
y_2^*,y_2^*)$ or $(y_2^*, y_2^*,y_2^*)$ is an efficient solution. If
$(0, y_2^*,y_2^*)$ is efficient, then we can use a simpler version of
the arguments in Claims~1 and~2 to prove that either $(0,0,0)$ or
$(0,1,1)$ is efficient. If $(y_2^*, y_2^*,y_2^*)$ is efficient, then
by applying Lemma~\ref{lem:opt-sum-min} with $k=3$, we deduce that
either $(0,0,0)$ or $(1,1,1)$ is efficient. Thus, there is a $0$-$1$
efficient solution.  The efficiency of Mechanism~\ref{mech:sum-min}
follows.
\end{proof}

When $k=2$ (resp., $1$), we can add one (resp., two) dummy facilities and use Theorem~\ref{thm:opt-sum-min} to establish that Mechanism~\ref{mech:sum-min} is efficient for $k = 2$ (resp., $1$). Theorem~\ref{thm:sum-min-sgsp-lower-bound} below provides a lower bound on the
approximation ratio of any SGSP efficient mechanism; this result implies that Mechanism~\ref{mech:sum-min} is not SGSP.

\begin{restatable}{theorem}{summinsgsplowerbound}
\label{thm:sum-min-sgsp-lower-bound}
Let $M$ be a single-facility SGSP $\alpha$-efficient $\OOFLG$ mechanism for some
positive constant $\alpha$.  Then $\alpha\geq 5/3$.
\end{restatable}
\begin{proof}
Let $n$ be a large integer.  We construct three $3n$-agent
single-facility $\OOFLG$ instances $(I,I)$, $(I,I')$, and $(I,I'')$.
In $(I,I)$, $(I,I')$, and $(I,I'')$, agent $1$ is located at $0$ and
dislikes $\{F_1\}$, agent $2$ is located at $1$ and dislikes
$\{F_1\}$, $n$ agents are located at $1/2$ and dislike $\{F_1\}$,
$n-1$ agents forming a set $U$ are located at $0$ and dislike
$\emptyset$, and $n-1$ agents forming a set $V$ are located at $1$ and
dislike $\emptyset$.  In $I$, all agents report truthfully. In $I'$
(resp. $I''$), all agents in $U$ (resp., $V$) report $\{F_1\}$ and the
remaining agents report truthfully.

Let the maximum social welfare for instances $I'$ and $I''$ be $\OPT'$
and $\OPT''$, respectively.  It is easy to see that
$\OPT'=\frac{5n}{2} - 1$ is achieved by building $F_1$ at $1$ on $I'$.
Likewise, $\OPT''=\frac{5n}{2} - 1$ is achieved by building $F_1$ at
$0$ on $I''$.  Let the social welfare achieved by mechanism $M$ on
$I'$ (resp., $I''$) be $\ALG'$ (resp., $\ALG''$).

Let $M$ build $F_1$ at $y$ (resp., $y'$, $y''$) on $I$ (resp., $I'$,
$I''$). We claim that $y = y' = y''$. To prove the claim, assume for
the sake of contradiction that $y \neq y'$. We consider two cases. If
$y < y'$, then agent $1$ benefits by forming a coalition with $V$ in
$(I, I')$. Similarly, if $y > y'$, then agent $2$ benefits by forming
a coalition with $U$ in $(I, I')$. Thus $y = y'$. Using a similar
argument for $(I, I'')$, we deduce that $y = y''$. Thus the claim
holds. We now consider two cases.

Case~1: $y \le 1/2$. Since $y' = y$, we have $\ALG' = \frac{3n}{2} -
y$.  Moreover, since $y \le 1/2$, we have $\ALG' \le 3n/2$.  Using
$\OPT' = \frac{5n}{2} - 1$ and $\ALG' \le 3n/2$, we obtain
\begin{equation*}
    \label{eqn:alpha-lower-bound-sgsp-utilitarian}
        \alpha \ge \frac{\frac{5n}{2} - 1}{\frac{3n}{2}} = \frac{5}{3}-\frac{2}{3n}.
\end{equation*}
Case~2: $y \ge 1/2$. Using similar arguments as in Case~1, but now for
$y'', \OPT''$, and $\ALG''$, we again find that
$\alpha\ge\frac{5}{3}-\frac{2}{3n}$.

Thus, in all cases, $\alpha$ is at least
$\frac{5}{3}-\frac{2}{3n}$. Since this bound approaches $\frac{5}{3}$
as $n$ tends to infinity, the theorem follows.
\end{proof}

In view of Theorem~\ref{thm:sum-min-sgsp-lower-bound}, it is natural
to try to determine the minimum value of $\alpha$ for which an SGSP
$\alpha$-efficient $\OOFLG$ mechanism exists. Below we present a
  $2$-efficient SGSP mechanism for arbitrary $k$.
  Section~\ref{subsubsec:sgsp-sum-min-interval} presents a
  $5/3$-efficient SGSP mechanism for $k=1$.  For $k>1$, it
remains an interesting open problem to improve the approximation ratio
of $2$, or to establish a tighter lower bound for the approximation
ratio.

\begin{mechanism}
\label{mech:sum-min-sgsp}
Let $\Instance$ denote the reported $\DOFL$ instance. Build all of the
facilities at $0$ if $\sum_{i\in [n]} x_i \ge \sum_{i \in [n]}
(1-x_i)$; otherwise, build all of the facilities at $1$.
\end{mechanism}

\begin{theorem}
\label{thm:sp-sum-min-sgsp}
Mechanism~\ref{mech:sum-min-sgsp} is SGSP.
\end{theorem}
\begin{proof}
Reported dislikes do not affect the locations at which the facilities are built. Hence the theorem follows.
\end{proof}

\begin{restatable}{theorem}{optsumminsgsp}
\label{thm:opt-sum-min-sgsp}
Mechanism~\ref{mech:sum-min-sgsp} is 2-efficient.
\end{restatable}
\begin{proof}
Let $I = \Instance$ denote the reported $\DOFL$ instance. Let $\ALG$ denote the social
welfare obtained by Mechanism~\ref{mech:sum-min-sgsp} on this instance, and let $\OPT$ denote the maximum possible social welfare on this instance.
We need to prove that
$2\cdot\ALG\geq\OPT$.

Assume without loss of generality that
Mechanism~\ref{mech:sum-min-sgsp} builds all of the facilities at
$0$. (A symmetric argument can be used in the case where all facilities
are built at $1$.) Then the welfare of an agent $i$ not in
$\indiff(I)$ is $x_i$ and the welfare of an agent $i$ in $\indiff(I)$
is $\max(x_{i}, 1-x_{i}) \ge x_{i}$. Thus $\ALG \ge \sum_{i\in [n]}
x_i$. As Mechanism~\ref{mech:sum-min-sgsp} builds the facilities at
$0$ and not $1$, we have $\sum_{i\in [n]}x_i\ge\sum_{i\in [n]}(1-x_i)$,
which implies that $\sum_{i\in [n]} x_i \ge n/2$. Combining the above
two inequalities, we obtain $\ALG \ge n/2$.  Since no agent has
welfare greater than $1$, we have $n \ge \OPT$.  Thus $2 \cdot \ALG
\ge n \ge \OPT$, as required.
\end{proof}

We now establish that the analysis of Theorem~\ref{thm:opt-sum-min-sgsp} is
tight by exhibiting a two-facility $\DOFL$ instance on which Mechanism~\ref{mech:sum-min-sgsp} achieves 
half of the optimal social welfare. 
For the reported $\DOFL$ instance $I = (2, 2, (0, 1), (\{F_1\}, \{F_2\}))$, it is easy to verify that the optimal social welfare is
$\SSW(I, (1,0)) = 2$, while the social welfare obtained by
Mechanism~\ref{mech:sum-min-sgsp} is $\SSW (I, (0,0))=1$.

\newcommand{\Reals}{\mathbb{R}}
\newcommand{\RealsPositive}{\Reals_{>0}}

\newcommand{\Loc}{x}
\newcommand{\LocFac}{y}
\newcommand{\Wt}{\gamma}

\newcommand{\Dist}{D}
\newcommand{\Weight}[1]{\Gamma(#1)}
\newcommand{\WtdSum}[1]{h(#1)}
\newcommand{\Prefix}[2]{\Id{prefix}(#1,#2)}
\newcommand{\Suffix}[2]{\Id{suffix}(#1,#2)}
\newcommand{\ObjFcn}{\Phi}
\newcommand{\Obj}[3]{\ObjFcn(#1,#2,#3)}
\newcommand{\ObjRatFcn}{\Psi}
\newcommand{\ObjRat}[4]{\ObjRatFcn(#1,#2,#3,#4)}

\newcommand{\RatSym}{\beta}
\newcommand{\RatOne}[1]{\RatSym(#1)}
\newcommand{\RatTwo}[2]{\RatSym(#1,#2)}

\newcommand{\Mech}{M}

\newcommand{\Id}[1]{\mbox{#1}}

\subsubsection{SGSP $5/3$-efficient mechanism}
\label{subsubsec:sgsp-sum-min-interval}

In this section, we design an SGSP $5/3$-efficient mechanism for
  the single-facility case. Throughout this section, we find it
  technically convenient to work over the interval $[-1,1]$ instead of
  $[0,1]$, and to allow the number of agents at a given location to be
  fractional. (We emphasize that the upper bound established in this
  section also holds for the special case in which the number of
  agents at any location is required to be an integer.)  We begin by
  introducing some useful definitions; these definitions will only be
  used in the present section.

We define a distribution as a finite subset $\Dist$ of
$[-1,1]\times\RealsPositive$ where no two pairs in $\Dist$ share the
same first component.  For any distribution $\Dist$, we define
$\Dist_{>0}$ as $\{(\Loc,\Wt)\mid \Loc>0\}$. Related expressions such
as $\Dist_{=0}$ are defined similarly.  For any distribution $\Dist$,
we define $\Weight{\Dist}$ as $\sum_{(\Loc,\Wt)\in\Dist}\Wt$ and
$\WtdSum{\Dist}$ as $\sum_{(\Loc,\Wt)\in\Dist}\Wt\Loc$.

We say that distribution $\Dist'$ is dominated by a distribution
$\Dist$ if for any pair $(\Loc,\Wt')$ in $\Dist'$, there is a pair
$(\Loc,\Wt)$ in $\Dist$ with $\Wt\geq\Wt'$. For any distributions
$\Dist$ and $\Dist'$ such that $\Dist$ dominates $\Dist'$, and any
$\LocFac$ in $[-1,1]$, we define $\Obj{\Dist}{\Dist'}{\LocFac}$ as
\[
\sum_{(\Loc,\Wt)\in\Dist}|\LocFac-\Loc|(\Wt-\Wt')+(1+|\Loc|)\Wt'
\]
where $\Wt'$ denotes $\Weight{\Dist'_{=\Loc}}$.

Let us elaborate on the intended interpretation of the function
$\Obj{\Dist}{\Dist'}{\LocFac}$.  The distribution $\Dist$ encodes the
number of agents (which we allow to be fractional) at each location in
$[-1,1]$ with a nonzero number of agents: If the pair $(\Loc,\Wt)$
belongs to $\Dist$, then there are $\Wt>0$ agents at location
$\Loc$. The distribution $\Dist'$ encodes the preferences of the
agents, in the following sense: If the pair $(\Loc,\Wt')$ belongs to
$\Dist'$, then there is a pair of the form $(\Loc,\Wt)$ in $\Dist$
such that $\Wt'\leq\Wt$ (since $\Dist$ dominates $\Dist'$), and we
understand that $\Wt'$ agents at $\Loc$ are indifferent to the
facility and $\Wt-\Wt'$ agents at $\Loc$ dislike the facility.
The location $\LocFac$ represents the location of the facility.  The
value of the function $\Obj{\Dist}{\Dist'}{\LocFac}$ represents the
total welfare of the agents, where the welfare of an agent at location
$\Loc$ is $|\LocFac-\Loc|$ (i.e., the distance between the agent and
the facility) if the agent dislikes the facility, and is
$(1+|\Loc|)$ (i.e., the distance to the farthest endpoint of the
interval $[-1,1]$) otherwise.

As the reader may recognize, there is a close connection between the
function $\Obj{\Dist}{\Dist'}{\LocFac}$ used in the present section
and the function $\SSW(I,\LocFac)$ used elsewhere in the
paper. Specifically, if we let $I$ denote a $\DOFL$ instance, and we encode the locations and
preferences of the agents in $I$ using distributions $\Dist$ and
$\Dist'$ as in the previous paragraph, then $\SSW(I,\LocFac)$ is equal
to $\Obj{\Dist}{\Dist'}{\LocFac}$.

For any distributions $\Dist$ and $\Dist'$ such that $\Dist$ dominates
$\Dist'$, any $\LocFac$ in $\{-1,0,1\}$, and any $\LocFac'$ in $[-1,1]$,
we define $\ObjRat{\Dist}{\Dist'}{\LocFac}{\LocFac'}$ as
\[
\frac{\Obj{\Dist}{\Dist'}{\LocFac'}}{\Obj{\Dist}{\Dist'}{\LocFac}}.
\]

For any distribution $\Dist$ and any location $\LocFac$ in
$\{-1,0,1\}$, we define $\RatTwo{\Dist}{\LocFac}$ as the maximum, over
all distributions $\Dist'$ dominated by $\Dist$ and all locations
$\LocFac'$ in $[-1,1]$, of
$\ObjRat{\Dist}{\Dist'}{\LocFac}{\LocFac'}$.  (Remark: If
$\Dist=\emptyset$, we consider the above ratio to be equal to $1$.) 

A $\OOFLG$ mechanism $\Mech$ maps any given distribution $\Dist$
  to a location $\Mech(\Dist)$ in $[-1,1]$.
Observe that a $\OOFLG $ mechanism $M$ is $\alpha$-efficient if and only if 
  $\beta (D,M(D)) \leq \alpha$ for all distributions $D$.  

\begin{custommech}{4}
  \label{mech:sum-min-fivethird-sgsp-interval}
  Let $I = (n, 1, \Location, \Aversion)$ denote the reported $\DOFL$
  instance.  We construct a corresponding distribution $\Dist$
  as follows: For each location $\Loc$ with one or more agents, we
  include the pair $(\Loc,\Wt)$ in $\Dist$ where $\Wt$ denotes the
  number of agents at $\Loc$. We then build the facility $F_1$ at the
  lexicographically least location in
  $\argmin_{\LocFac\in\{-1,0,1\}}\RatTwo{\Dist}{\LocFac}$.
\end{custommech}

Mechanism~\ref{mech:sum-min-fivethird-sgsp-interval} is SGSP because it disregards the reported aversion profile.
In the remainder of this section, we establish two other key
properties of Mechanism~\ref{mech:sum-min-fivethird-sgsp-interval}.
First, in Theorem~\ref{thm:fiveThirdsFast} below, we establish that
Mechanism~\ref{mech:sum-min-fivethird-sgsp-interval} admits a fast
implementation. Second, in Theorem~\ref{thm:fiveThirds} below, we show
that Mechanism~\ref{mech:sum-min-fivethird-sgsp-interval} achieves an
approximation ratio of $5/3$. Given the $5/3$ lower bound established
in Theorem~\ref{thm:sum-min-sgsp-lower-bound}, this approximation
ratio is optimal for any SGSP mechanism.

Lemma~\ref{lem:fiveThirdsAdversaryHelper} below is a weighted
generalization of Lemma~\ref{lem:1-maxian}, and can be
justified in a similar manner. For the sake of completeness, and also to
exercise some of the notations introduced above, we provide an alternate
proof below.

\begin{lemma}
\label{lem:fiveThirdsAdversaryHelper}
Let $\Dist$ be a distribution. Then
\[
\max_{\LocFac\in[-1,1]}\Obj{\Dist}{\emptyset}{\LocFac}
=
\max_{\LocFac\in\{-1,1\}}\Obj{\Dist}{\emptyset}{\LocFac}.
\]
\end{lemma}
\begin{proof}
Let $\LocFac$ belong to $[-1,1]$. We need to prove that
$\Obj{\Dist}{\emptyset}{\LocFac}$ is at most
$\max(\Obj{\Dist}{\emptyset}{-1},\Obj{\Dist}{\emptyset}{1})$.  We
consider two cases.

Case~1: $\Weight{\Dist_{\leq\LocFac}}\geq\Weight{\Dist}/2$. Observe that
\begin{eqnarray*}
\Obj{\Dist_{\leq\LocFac}}{\emptyset}{1}
-\Obj{\Dist_{\leq\LocFac}}{\emptyset}{\LocFac}
& = &
\sum_{(\Loc,\Wt)\in\Dist_{\leq\LocFac}}[(1-\Loc)-(\LocFac-\Loc)]\Wt\\
& = &
(1-\LocFac)\Weight{\Dist_{\leq\LocFac}},
\end{eqnarray*}
$\Obj{\Dist_{>\LocFac}}{\emptyset}{1}\geq 0$, and
\begin{eqnarray*}
\Obj{\Dist_{>\LocFac}}{\emptyset}{\LocFac}
& = &
\sum_{(\Loc,\Wt)\in\Dist_{>\LocFac}}(\Loc-\LocFac)\Wt\\
& \leq &
(1-\LocFac)\Weight{\Dist_{>\LocFac}}.
\end{eqnarray*}
Combining the above inequalities with $\Weight{\Dist_{\leq\LocFac}}\geq\Weight{\Dist}/2\geq\Weight{\Dist_{>\LocFac}}$, we obtain
\begin{eqnarray*}
& & \Obj{\Dist}{\emptyset}{1}-\Obj{\Dist}{\emptyset}{\LocFac}\\
& = &
\Obj{\Dist_{\leq\LocFac}}{\emptyset}{1}
-\Obj{\Dist_{\leq\LocFac}}{\emptyset}{\LocFac}
+\Obj{\Dist_{>\LocFac}}{\emptyset}{1}
-\Obj{\Dist_{>\LocFac}}{\emptyset}{\LocFac}\\
& \geq &
(1-\LocFac)\Weight{\Dist_{\leq\LocFac}}
+0-(1-\LocFac)\Weight{\Dist_{>\LocFac}}\\
& \geq & 0.
\end{eqnarray*}

Case~2: $\Weight{\Dist_{\geq\LocFac}}\geq\Weight{\Dist}/2$. Using an
argument symmetric to that used in Case~1, we find that
$\Obj{\Dist}{\emptyset}{\LocFac}\leq\Obj{\Dist}{\emptyset}{-1}$.
\end{proof}

Lemma~\ref{lem:fiveThirdsAdversary} below shows that given the
locations and preferences of the agents, the optimal location for the
facility is at one of the endpoints of the interval $[-1,1]$.

\begin{lemma}
\label{lem:fiveThirdsAdversary}
Let $\Dist$ and $\Dist'$ be distributions such that $\Dist$ dominates
$\Dist'$. Then
\[
\max_{\LocFac\in[-1,1]}\Obj{\Dist}{\Dist'}{\LocFac}
=
\max_{\LocFac\in\{-1,1\}}\Obj{\Dist}{\Dist'}{\LocFac}.
\]
\end{lemma}
\begin{proof}
Let $\LocFac$ belong to $[-1,1]$ and let $\Dist^*$ denote the unique
distribution such that
$\Weight{\Dist^*_{=\Loc}}=\Weight{\Dist_{=\Loc}}-\Weight{\Dist'_{=\Loc}}$
for all $\Loc$ in $[-1,1]$. Observe that
$\Obj{\Dist}{\Dist'}{\LocFac}$ is equal to
$\Obj{\Dist^*}{\emptyset}{\LocFac}+\Obj{\Dist'}{\Dist'}{\LocFac}$.
Since $\Obj{\Dist'}{\Dist'}{\LocFac}$ is equal to
$\sum_{(\Loc,\Wt)\in\Dist'}(1+|\Loc|)\Wt$, which is independent of
$\LocFac$, and since Lemma~\ref{lem:fiveThirdsAdversaryHelper} implies
that $\Obj{\Dist^*}{\emptyset}{\LocFac}$ is at most
$\max_{\LocFac'\in\{-1,1\}}\Obj{\Dist^*}{\emptyset}{\LocFac'}$, we
deduce that $\Obj{\Dist}{\Dist'}{\LocFac}$ is at most
$\max_{\LocFac'\in\{-1,1\}}\Obj{\Dist}{\Dist'}{\LocFac'}$.
\end{proof}

For any distributions $\Dist$ and $\Dist'$, we define
$\min(\Dist,\Dist')$ as the unique distribution $\Dist^*$ such that
\[
\Weight{\Dist^*_{=\Loc}}=\min(\Weight{\Dist_{=\Loc}},\Weight{\Dist'_{=\Loc}})
\]
for all $\Loc$ in $[-1,1]$. Lemma~\ref{lem:fiveThirdsEval} below gives
a useful way to rewrite the expression $\Obj{\Dist}{\Dist'}{\LocFac}$
for all $\LocFac$ in $\{-1,0,1\}$.

\begin{lemma}
\label{lem:fiveThirdsEval}
Let $\Dist$ and $\Dist'$ be distributions such that $\Dist$ dominates
$\Dist'$. Then $\Obj{\Dist}{\Dist'}{-1}$ is equal to
$\Weight{\Dist}+\WtdSum{\Dist}-2\WtdSum{\min(\Dist',\Dist_{<0})}$,
$\Obj{\Dist}{\Dist'}{0}$ is equal to
$\Weight{\Dist'}-\WtdSum{\Dist_{<0}}+\WtdSum{\Dist_{>0}}$, and
$\Obj{\Dist}{\Dist'}{1}$ is equal to
$\Weight{\Dist}-\WtdSum{\Dist}+2\WtdSum{\min(\Dist',\Dist_{>0})}$.
\end{lemma}
\begin{proof}
Let $(\Loc,\Wt)$ belong to $\Dist$ and let $\Wt'$ denote
$\Weight{\Dist'_{=\Loc}}$.

The contribution of the pair $(\Loc,\Wt)$ to $\Obj{\Dist}{\Dist'}{-1}$
is $\Wt'(1-\Loc)+(\Wt-\Wt')(1+\Loc)=\Wt+\Wt\Loc-2\Wt'\Loc$ if $\Loc$
belongs to $\Dist_{<0}$, and is $\Wt+\Wt\Loc$ otherwise. The first
claim of the lemma follows. A symmetric argument establishes the third
claim.

The contribution of the pair $(\Loc,\Wt)$ to $\Obj{\Dist}{\Dist'}{0}$
is $\Wt'(1-\Loc)-(\Wt-\Wt')\Loc=\Wt'-\Wt\Loc$ if $\Loc$ belongs to
$\Dist_{<0}$, and is $\Wt'(1+\Loc)+(\Wt-\Wt')\Loc=\Wt'+\Wt\Loc$
otherwise. The second claim follows.
\end{proof}

Given the locations of the agents, but not their preferences,
Lemma~\ref{lem:fiveThirdsPlusMinusA} below characterizes the best
possible approximation ratio that can be guaranteed by locating the
facility at $-1$ (resp., $1$).

\begin{lemma}
\label{lem:fiveThirdsPlusMinusA}
Let $\Dist$ be a distribution and let $\LocFac$ belong to $\{-1,1\}$.
Then
\[
\RatTwo{\Dist}{\LocFac}
=
\ObjRat{\Dist}{\Dist'}{\LocFac}{-\LocFac}
\]
where $\Dist'$ denotes $\Dist_{>0}$ (resp., $\Dist_{<0}$) if
$\LocFac$ is equal to $-1$ (resp., $1$).
\end{lemma}
\begin{proof}
By symmetry, it is sufficient to consider the case where $\LocFac$ is
equal to $1$.  Lemma~\ref{lem:fiveThirdsAdversary} implies that
$\RatTwo{\Dist}{1}$ is equal to the maximum, over all distributions
$\Dist'$ dominated by $\Dist$, of $\max_{\LocFac'\in\{-1,1\}}
\ObjRat{\Dist}{\Dist'}{1}{\LocFac'}$.
Lemma~\ref{lem:fiveThirdsEval} implies that setting
$\Dist'$ to $\Dist_{<0}$ simultaneously maximizes
$\Obj{\Dist}{\Dist'}{-1}$ and minimizes
$\Obj{\Dist}{\Dist'}{1}$. Thus, setting $\Dist'$ to $\Dist_{<0}$
maximizes $\ObjRat{\Dist}{\Dist'}{1}{-1}$. Moreover, it is easy to see
that $\ObjRat{\Dist}{\Dist_{<0}}{1}{-1}\geq 1$, and that
$\ObjRat{\Dist}{\Dist'}{1}{1}=1$ for all $\Dist'$ dominated by
$\Dist$.  The claim of the lemma follows.
\end{proof}

Given the locations of the agents, but not their preferences,
Lemma~\ref{lem:fiveThirdsPlusMinusB} gives a useful way to rewrite the
best possible approximation ratio that can be guaranteed by locating
the facility in $\{-1, 1\}$.

\begin{lemma}
\label{lem:fiveThirdsPlusMinusB}
Let $\Dist$ be a distribution. Then
\[
\min_{\LocFac\in\{-1,1\}}\RatTwo{\Dist}{\LocFac}
=\frac{\Weight{\Dist}-\WtdSum{\Dist_{<0}}+\WtdSum{\Dist_{>0}}}
{\Weight{\Dist}+|\WtdSum{\Dist}|}.
\]
\end{lemma}
\begin{proof}
Follows straightforwardly from Lemmas~\ref{lem:fiveThirdsEval}
and~\ref{lem:fiveThirdsPlusMinusA}.
\end{proof}

We will make repeated use of the following simple fact, so we state it
explicitly.

\begin{fact}
\label{fact:fiveThirdsMono}
Let $f(x)$ denote $\frac{x+c}{x+1}$ where $c$ is a positive constant.
Then $c>1$ (resp., $c<1$, $c=1$) implies $f(x_0)>f(x_1)$ (resp.,
$f(x_0)<f(x_1)$, $f(x_0)=f(x_1)$) for all $x_0$ and $x_1$ such that
$0\leq x_0<x_1$.
\end{fact}

For any distribution $\Dist$ and any integer $k$ in
$\{0,\ldots,|\Dist|\}$, we define $\Prefix{\Dist}{k}$ (resp.,
$\Suffix{\Dist}{k}$) as the subset of $\Dist$ consisting of the $k$
lexicographically smallest (resp., largest) pairs.  Given the
locations of the agents, but not their preferences,
Lemma~\ref{lem:fiveThirdsPrefixSuffix} below characterizes the best
possible approximation ratio that can be guaranteed by locating the
facility at $0$.

\begin{lemma}
\label{lem:fiveThirdsPrefixSuffix}
Let $\Dist$ be a distribution. Then
$\RatTwo{\Dist}{0}$ is equal to
\[
\max\left(\max_{k\in\{0,\ldots,|\Dist_{<0}|\}}
\ObjRat{\Dist}{\Prefix{\Dist}{k}}{0}{-1},
\max_{k\in\{0,\ldots,|\Dist_{>0}|\}}
\ObjRat{\Dist}{\Suffix{\Dist}{k}}{0}{1}\right).
\]
\end{lemma}
\begin{proof}
Lemma~\ref{lem:fiveThirdsAdversary} implies that $\RatTwo{\Dist}{0}$
is equal to the maximum, over all distributions $\Dist'$ dominated by
$\Dist$, of
$\max_{\LocFac\in\{-1,1\}}\ObjRat{\Dist}{\Dist'}{0}{\LocFac}$.  Let
$\Dist'$ be a distribution dominated by $\Dist$ that maximizes
$\ObjRat{\Dist}{\Dist'}{0}{-1}$, and let $\xi$ denote
$\ObjRat{\Dist}{\Dist'}{0}{-1}$.

Claim~1: $\Dist'_{\geq 0}=\emptyset$. Assume for the sake of
contradiction that $\Dist'_{\geq
  0}\not=\emptyset$. Lemma~\ref{lem:fiveThirdsEval} implies that
$\Obj{\Dist}{\Dist'_{<0}}{-1}=\Obj{\Dist}{\Dist'}{-1}$ and
$\Obj{\Dist}{\Dist'_{<0}}{0}<\Obj{\Dist}{\Dist'}{0}$. Hence
$\ObjRat{\Dist}{\Dist'_{<0}}{0}{-1}>\xi$, a contradiction.

Claim~2: For any $\Loc$ and $\Loc'$ such that $-1\leq\Loc<\Loc'<0$ and
$\Weight{\Dist'_{=\Loc'}}>0$, we have
$\Weight{\Dist'_{=\Loc}}=\Weight{\Dist_{=\Loc}}$. Assume for the sake
of contradiction that $-1\leq\Loc<\Loc'<0$,
$\Weight{\Dist'_{=\Loc'}}>0$, and
$\Weight{\Dist'_{=\Loc}}<\Weight{\Dist_{=\Loc}}$. Let $\delta$ denote
$\min(\Weight{\Dist'_{=\Loc'}},\Weight{\Dist_{=\Loc}}-\Weight{\Dist'_{=\Loc}})$.
Thus $\delta>0$. Let $\Dist''$ denote the distribution
\[
(\Dist'\setminus
\{(\Loc,\Weight{\Dist'_{=\Loc}}),(\Loc',\Weight{\Dist'_{=\Loc'}})\})
\cup
\{(\Loc,\Weight{\Dist'_{=\Loc}}+\delta),(\Loc',\Weight{\Dist'_{=\Loc'}}-\delta)\})
\]
Note that $\Dist''$ is dominated by $\Dist$.
Lemma~\ref{lem:fiveThirdsEval} implies that
$\Obj{\Dist}{\Dist''}{-1}>\Obj{\Dist}{\Dist'}{-1}$ and
$\Obj{\Dist}{\Dist''}{0}=\Obj{\Dist}{\Dist'}{0}$. Hence
$\ObjRat{\Dist}{\Dist''}{0}{-1}>\xi$, a contradiction.

Let $k^*$ denote $|\Dist'|$, which is at most $|\Dist_{<0}|$ by
Claim~1.

Claim~3: If $k^*>0$ then
\[
\max(\ObjRat{\Dist}{\Prefix{\Dist}{k^*-1}}{0}{-1},
\ObjRat{\Dist}{\Prefix{\Dist}{k^*}}{0}{-1})=\xi.
\]
Let $(\Loc,\Wt)$ denote the lexicographically greatest pair in
$\Prefix{\Dist}{k^*}$, let $\Dist^{(0)}$ denote
$\Prefix{\Dist}{k^*-1}$, and for any $t$ such that $0<t\leq 1$, let
$\Dist^{(t)}$ denote $\Dist^{(0)}+(\Loc,t\Wt)$.  For any $t$ in
$[0,1]$, let $g(t)$ denote $\ObjRat{\Dist}{\Dist^{(t)}}{0}{-1}$. Using
Fact~\ref{fact:fiveThirdsMono}, it is straightforward to prove that
$g(t)$ is monotonic over the unit interval and hence $\max_{t\in
  [0,1]}g(t)=\max(g(0),g(1))$. Claim~3 follows.

Combining Claim~3 with the observation that $k^*=0$ implies
$\Dist'=\Prefix{\Dist}{0}$, we deduce that there is an integer $k$ in
$\{0,\ldots,|\Dist_{<0}|\}$ such that
$\ObjRat{\Dist}{\Prefix{\Dist}{k}}{0}{-1})=\xi$.  
In other words, there is an integer $k$ in $\{0,\ldots ,|D_{<0}|\}$ 
such that $\ObjRat{\Dist}{\Dist'}{0}{-1})$ is maximized by setting $D'$ to $\Prefix{\Dist}{k}$.

Using an entirely
symmetric argument, we find that there is an integer $k$ in
$\{0,\ldots,|\Dist_{>0}|\}$ such that $\ObjRat{\Dist}{\Dist'}{0}{1})$
is maximized by setting $\Dist'$ to $\Suffix{\Dist}{k}$. The claim of
the lemma follows.
\end{proof}

We are now ready to describe a fast implementation of
Mechanism~\ref{mech:sum-min-fivethird-sgsp-interval}.

\begin{theorem}
\label{thm:fiveThirdsFast}
There is an 
$O(|\Dist|\log|\Dist|)$-time implementation of Mechanism~\ref{mech:sum-min-fivethird-sgsp-interval}.
\end{theorem}
\begin{proof}
Let $\Dist$ be a given distribution. We begin by spending
$O(|\Dist|\log|\Dist|)$ time to sort the pairs of $\Dist$ in
lexicographic order.  Using Lemma~\ref{lem:fiveThirdsPlusMinusA}, we
can then compute $\RatTwo{\Dist}{-1}$ and $\RatTwo{\Dist}{1}$ in
$O(|\Dist|)$ time.  Likewise, using
Lemma~\ref{lem:fiveThirdsPrefixSuffix}, we can compute
$\RatTwo{\Dist}{0}$ in $O(|\Dist|)$ time. The theorem follows.
\end{proof}

For any distribution $\Dist$, let $\Dist_-$ (resp., $\Dist_+$) denote
$\Prefix{\Dist}{k}$ (resp., $\Suffix{\Dist}{k}$) where $k$ is the
least integer in $\{0,\ldots,|\Dist_{<0}|\}$ (resp.,
$\{0,\ldots,|\Dist_{>0}|\}$) maximizing the first (resp., second)
$\max$ expression in the statement of
Lemma~\ref{lem:fiveThirdsPrefixSuffix}.
  
We say that a distribution $\Dist$ is \emph{trivial} if
$\Dist_{<0}=\emptyset$ or $\Dist_{>0}=\emptyset$.  We say that a
distribution $\Dist$ is \emph{special} if $|\Dist_{<0}|$ and
$|\Dist_{>0}|$ each belong to $\{1,2\}$.
Lemma~\ref{lem:fiveThirdsNontrivial} below allows us to bound the
performance of Mechanism~\ref{mech:sum-min-fivethird-sgsp-interval} on
any nontrivial distribution of agent locations in terms of its
performance on any special distribution of agent locations.

\begin{lemma}
\label{lem:fiveThirdsNontrivial}
Let $\Dist$ be a nontrivial distribution. Then there exists a special
distribution $\Dist'$ such that
\[
\RatTwo{\Dist}{\LocFac}=\RatTwo{\Dist'}{\LocFac}
\]
for all locations $\LocFac$ in $\{-1,0,1\}$.
\end{lemma}
\begin{proof}
We begin by introducing some useful notation. For any two
distributions $\Dist$ and $\Dist'$, we write $\Dist\sim\Dist'$ to mean
that $\Weight{\Dist}=\Weight{\Dist'}$ and
$\WtdSum{\Dist}=\WtdSum{\Dist'}$, and we write $\Dist\cong\Dist'$ to
mean that $D\sim D'$, $D_{<0}\sim D'_{<0}$, $D_{>0}\sim D'_{>0}$,
$D_-\sim D'_-$, and $D_+\sim D'_+$.

Using Lemmas~\ref{lem:fiveThirdsPlusMinusB}
and~\ref{lem:fiveThirdsPrefixSuffix}, it is straightforward to prove
that for any distributions $\Dist$ and $\Dist'$ such that
$\Dist\cong\Dist'$, we have
$\RatTwo{\Dist}{\LocFac}=\RatTwo{\Dist'}{\LocFac}$ for all locations
$\LocFac$ in $\{-1,0,1\}$. Accordingly, it is sufficient to show how
to map any given nontrivial distribution $\Dist$ to a special
distribution $\Dist'$ such that $\Dist\cong\Dist'$.

Let $\Dist$ be a nontrivial distribution, let $\Dist^{(-2)}$ denote
$\Dist_-$, let $\Dist^{(-1)}$ denote $\Dist_{<0}\setminus\Dist_-$, let
$\Dist^{(0)}$ denote $\Dist_{=0}$, let $\Dist^{(1)}$ denote
$\Dist_{>0}\setminus\Dist_+$, and let $\Dist^{(2)}$ denote $\Dist_+$.
For any $i$ in $\{-2,\ldots,2\}$, let $\Wt_i$ denote
$\Weight{\Dist^{(i)}}$, let $\Delta_i$ denote $\WtdSum{\Dist^{(i)}}$,
and let $\Dist'$ denote the distribution
\[
\{(\Delta_i/\Wt_i,\Wt_i)\mid i\in\{-2,\ldots,2\}\;\wedge\;\Wt_i>0\}.
\]
It is straightforward to argue that $\Dist'$ is special,
$\Dist'\sim\Dist$, $\Dist'_{<0}\sim\Dist_{<0}$, and
$\Dist'_{>0}\sim\Dist_{>0}$. To complete the proof, it remains to
argue that $\Dist'_-\sim\Dist_-$ and $\Dist'_+\sim \Dist_+$.  Claims~1
and~2 below imply that $\Dist'_-\sim\Dist_-$.  A symmetric argument
establishes that $\Dist'_+\sim\Dist_+$.

Claim~1: For any $k$ in $\{0,\ldots,|\Dist'_{<0}|\}$, there exists an
$\ell$ in $\{0,\ldots,|\Dist_{<0}|\}$ such that
$\Prefix{\Dist'}{k}\sim\Prefix{\Dist}{\ell}$; moreover, there exists a
$k$ in $\{0,\ldots,|\Dist'_{<0}|\}$ such that
$\Prefix{\Dist'}{k}\sim\Dist_-$.  To prove Claim~1, we consider three
cases.

Case~1: $\Dist_-=\emptyset$. In this case, we have $|\Dist'_{<0}|=1$,
$\Prefix{\Dist'}{0}=\emptyset=\Prefix{\Dist}{0}=\Dist_-$ and
$\Prefix{\Dist'}{1}\sim\Prefix{\Dist}{|\Dist_{<0}|}=\Dist_{<0}$.

Case~2: $\Dist_-\not=\emptyset$ and $|\Dist'_{<0}|=1$.  In this case,
we have $\Prefix{\Dist'}{0}=\emptyset=\Prefix{\Dist}{0}$ and
$\Prefix{\Dist'}{1}\sim\Prefix{\Dist}{|\Dist_{<0}|}=\Dist_{<0}=\Dist_-$.

Case~3: $\Dist_-\not=\emptyset$ and $|\Dist'_{<0}|=2$.  In this case,
we have $\Prefix{\Dist'}{0}=\emptyset=\Prefix{\Dist}{0}$,
$\Prefix{\Dist'}{1}\sim\Prefix{\Dist}{|\Dist_-|}=\Dist_-$, and
$\Prefix{\Dist'}{2}\sim\Prefix{\Dist}{|\Dist_{<0}|}=\Dist_{<0}$.

Claim~2: For any $k$ in $\{0,\ldots,|\Dist'_{<0}|\}$ and any $\ell$ in
$\{0,\ldots,|\Dist_{<0}|\}$ such that
$\Prefix{\Dist'}{k}\sim\Prefix{\Dist}{\ell}$, we have
\[
\ObjRat{\Dist'}{\Prefix{\Dist'}{k}}{0}{-1}
=
\ObjRat{\Dist}{\Prefix{\Dist}{\ell}}{-1}{0}.
\]
Claim~2 follows from Lemma~\ref{lem:fiveThirdsEval} since
$\Dist'\sim\Dist$, $\Dist'_{<0}\sim\Dist_{<0}$,
$\Dist'_{>0}\sim\Dist_{>0}$, and
$\Prefix{\Dist'}{k}\sim\Prefix{\Dist}{\ell}$.
\end{proof}

When the agent locations are described by a special distribution
$\Dist$ such that $\WtdSum{\Dist}>0$,
Lemma~\ref{lem:fiveThirdsBalanceHelper} below enables us to bound the
performance of Mechanism~\ref{mech:sum-min-fivethird-sgsp-interval} in
terms of its performance on a special distribution $\Dist'$ such that
either (1) $\WtdSum{\Dist'}=0$ or (2) $\WtdSum{\Dist'}>0$ and
$1=|\Dist'_{>0}|<|\Dist_{>0}|=2$. A symmetric claim holds for the case
where $\WtdSum{\Dist}<0$.

The main idea underlying the proof of
Lemma~\ref{lem:fiveThirdsBalanceHelper} is to create the desired
distribution $\Dist'$ from the given distribution $\Dist$ by sliding
the lexicographically least pair $(\Loc,\Wt)$ in $\Dist_{>0}$ towards
$0$ (i.e., replacing it with a pair $(\Loc',\Wt)$ where
$0\leq\Loc'<\Loc$) until either $0$ is reached or
$\WtdSum{\Dist'}=0$. Using Lemmas~\ref{lem:fiveThirdsPlusMinusB}
and~\ref{lem:fiveThirdsPrefixSuffix}, we are able to prove that the
best possible approximation ratio that can be guaranteed by
Mechanism~\ref{mech:sum-min-fivethird-sgsp-interval} under
distribution $\Dist'$ is higher than it is under distribution $\Dist$.

\begin{lemma}
\label{lem:fiveThirdsBalanceHelper}
Let $\Dist$ be a special distribution such that
$\WtdSum{\Dist}>0$. Then there exists a special distribution $\Dist'$
such that
\[
\min_{\LocFac\in\{-1,0,1\}}\RatTwo{\Dist'}{\LocFac}
>\min_{\LocFac\in\{-1,0,1\}}\RatTwo{\Dist}{\LocFac},
\]
$0\leq\WtdSum{\Dist'}<\WtdSum{\Dist}$, and if $\WtdSum{\Dist'}>0$ then
$1=|\Dist'_{>0}|<|\Dist_{>0}|=2$.
\end{lemma}
\begin{proof}
Let $(\Loc^-,\Wt^-)$ denote the lexicographically least pair in
$\Dist$, and let $(\Loc^+,\Wt^+)$ denote the lexicographically
greatest pair in $\Dist$.

Let $(\Loc^*,\Wt^*)$ denote the lexicographically least pair in
$\Dist_{>0}$. We remark that if $|\Dist_{>0}|=1$, then
$(\Loc^*,\Wt^*)=(\Loc^+,\Wt^+)$.

We define $\Loc'$ as $\max(\Loc^*-h(\Dist)/\Wt^*,0)$; thus
$0\leq\Loc'<\Loc^*$.  We now construct the desired distribution
$\Dist'$ as follows. If $\Loc'>0$, then we define $\Dist'$ as
$\Dist-(\Loc^*,\Wt^*)+(\Loc',\Wt^*)$. Otherwise, $\Loc'=0$ and we
define $\Dist'$ as
$\Dist_{\not=0}-(\Loc^*,\Wt^*)+(0,\Weight{\Dist_{=0}}+\Wt^*)$.  It is
easy to see that the following conditions hold: $\Dist'$ is special;
$0\leq\WtdSum{\Dist'}<\WtdSum{\Dist}$; if $\WtdSum{\Dist'}>0$ then
$1=|\Dist'_{>0}|<|\Dist_{>0}|=2$; $\Weight{\Dist'}=\Weight{\Dist}$;
$\Dist'_{<0}=\Dist_{<0}$;
$0\leq\WtdSum{\Dist'_{>0}}-\Wt^*\Loc'=\WtdSum{\Dist_{>0}}-\Wt^*\Loc^*$.

Claim~1: $\min_{\LocFac\in\{-1,1\}}\RatTwo{\Dist'}{\LocFac}
>\min_{\LocFac\in\{-1,1\}}\RatTwo{\Dist}{\LocFac}$.  Since
$\WtdSum{\Dist}>0$, Lemma~\ref{lem:fiveThirdsPlusMinusB} implies
\begin{eqnarray*}
  \min_{\LocFac\in\{-1,1\}}\RatTwo{\Dist}{\LocFac}
& = &
  \frac{\Weight{\Dist}-\WtdSum{\Dist_{<0}}+\WtdSum{\Dist_{>0}}}
  {\Weight{\Dist}+\WtdSum{\Dist_{<0}}+\WtdSum{\Dist_{>0}}}.
\end{eqnarray*}
Similarly, we have
\begin{eqnarray*}
  \min_{\LocFac\in\{-1,1\}}\RatTwo{\Dist'}{\LocFac}
& = &
  \frac{\Weight{\Dist'}-\WtdSum{\Dist'_{<0}}+\WtdSum{\Dist'_{>0}}}
  {\Weight{\Dist'}+\WtdSum{\Dist'_{<0}}+\WtdSum{\Dist'_{>0}}}.
\end{eqnarray*}
Since $\Weight{\Dist'}=\Weight{\Dist}$, $\Dist'_{<0}=\Dist_{<0}$,
$0\leq \Loc'<\Loc^*$,
$0\leq\WtdSum{\Dist'_{>0}}-\Wt^*\Loc'=\WtdSum{\Dist_{>0}}-\Wt^*\Loc^*$,
$\Wt^*>0$, and
$\Weight{\Dist}-\WtdSum{\Dist_{<0}}>\Weight{\Dist}+\WtdSum{\Dist_{<0}}>0$,
Fact~\ref{fact:fiveThirdsMono} implies that Claim~1 holds.

Let $X$ denote
$\max_{\LocFac\in\{-1,1\}}\ObjRat{\Dist}{\emptyset}{0}{\LocFac}$, let
$Y_{-1}$ denote $\ObjRat{\Dist}{\Dist_{<0}}{0}{-1}$, let $Y_1$ denote
$\ObjRat{\Dist}{\Dist_{>0}}{0}{1}$, let $Z_{-1}$ denote
$\ObjRat{\Dist}{\{(\Loc^-,\Wt^-)\}}{0}{-1}$, let $Z_1$ denote
$\ObjRat{\Dist}{\{(\Loc^+,\Wt^+)\}}{0}{1}$, and let $X'$, $Y_{-1}'$,
$Y_1'$, and $Z_{-1}'$ be defined similarly, except in terms of
$\Dist'$ instead of $\Dist$.  If $|\Dist_{>0}|=2$ then
$(\Loc^+,\Wt^+)$ belongs to $\Dist'$, and we likewise define $Z'_1$ as
$\ObjRat{\Dist'}{\{(\Loc^+,\Wt^+)\}}{0}{1}$.

Since $\Dist$ is special, we deduce from
Lemma~\ref{lem:fiveThirdsPrefixSuffix} that
\begin{equation*}
\RatTwo{\Dist}{0}=
\begin{dcases*}
\max(X,Y_{-1},Y_1,Z_{-1}) & if $|\Dist_{>0}|=1$\\
\max(X,Y_{-1},Y_1,Z_{-1},Z_1) & if $|\Dist_{>0}|=2$,
\end{dcases*}
\end{equation*}
and
\begin{equation*}
\label{eqn:xyz-star}
\RatTwo{\Dist'}{0}=
\begin{dcases*}
\max(X',Y'_{-1},Y'_1,Z'_{-1}) & if $|\Dist_{>0}|=1$\\
\max(X',Y'_{-1},Y'_1,Z'_{-1},Z'_1) & if $|\Dist_{>0}|=2$.
\end{dcases*}
\end{equation*}
Thus Claims~2 through~6 below imply that
$\RatTwo{\Dist}{0}\leq\RatTwo{\Dist'}{0}$.  Combining this inequality
with Claim~1, we find that $\Dist'$ satisfies the requirements of the
lemma.  It remains only to prove Claims~2 through~6.

Claim~2: $X'>X$. Since $\WtdSum{\Dist}>0$,
Lemma~\ref{lem:fiveThirdsEval} implies that
\begin{eqnarray*}
X & = & \frac{\Weight{\Dist}+\WtdSum{\Dist}}
{-\WtdSum{\Dist_{<0}}+\WtdSum{\Dist_{>0}}}\\
& = &
\frac{\Weight{\Dist}+\WtdSum{\Dist_{<0}}+\WtdSum{\Dist_{>0}}}
{-\WtdSum{\Dist_{<0}}+\WtdSum{\Dist_{>0}}}.
\end{eqnarray*}
Likewise, since $\WtdSum{\Dist'}\geq 0$,
Lemma~\ref{lem:fiveThirdsEval} implies that
\[
X'=
\frac{\Weight{\Dist'}+\WtdSum{\Dist'_{<0}}+\WtdSum{\Dist'_{>0}}}
{-\WtdSum{\Dist'_{<0}}+\WtdSum{\Dist'_{>0}}}.
\]
Since $\Weight{\Dist'}=\Weight{\Dist}$, $\Dist'_{<0}=\Dist_{<0}$,
$0\leq \Loc'<\Loc^*$,
$0\leq\WtdSum{\Dist'_{>0}}-\Wt^*\Loc'=\WtdSum{\Dist_{>0}}-\Wt^*\Loc^*$,
$\Wt^*>0$, and $\Weight{\Dist}+\WtdSum{\Dist_{<0}}\geq
\Weight{\Dist_{>0}}\geq\WtdSum{\Dist_{>0}}>-\WtdSum{\Dist_{<0}}>0$,
Fact~\ref{fact:fiveThirdsMono} implies that Claim~2 holds.

Claim~3: $Y'_{-1}>Y_{-1}$.
Lemma~\ref{lem:fiveThirdsEval} implies that
\[
Y_{-1} = \frac{\Weight{\Dist}-\WtdSum{\Dist_{<0}}+\WtdSum{\Dist_{>0}}}
{\Weight{\Dist_{<0}}-\WtdSum{\Dist_{<0}}+\WtdSum{\Dist_{>0}}}
\]
and
\[
Y'_{-1} = \frac{\Weight{\Dist'}-\WtdSum{\Dist'_{<0}}+\WtdSum{\Dist'_{>0}}}
{\Weight{\Dist'_{<0}}-\WtdSum{\Dist'_{<0}}+\WtdSum{\Dist'_{>0}}}.
\]
Since $\Weight{\Dist'}=\Weight{\Dist}$, $\Dist'_{<0}=\Dist_{<0}$,
$0\leq \Loc'<\Loc^*$,
$0\leq\WtdSum{\Dist'_{>0}}-\Wt^*\Loc'=\WtdSum{\Dist_{>0}}-\Wt^*\Loc^*$,
$\Wt^*>0$, and
$\Weight{\Dist}-\WtdSum{\Dist_{<0}}>\Weight{\Dist_{<0}}-\WtdSum{\Dist_{<0}}>0$,
Fact~\ref{fact:fiveThirdsMono} implies that Claim~3 holds.

Claim~4: $Y'_1>Y_1$.
Lemma~\ref{lem:fiveThirdsEval} implies that
\[
Y_1 = \frac{\Weight{\Dist}-\WtdSum{\Dist_{<0}}+\WtdSum{\Dist_{>0}}}
{\Weight{\Dist_{>0}}-\WtdSum{\Dist_{<0}}+\WtdSum{\Dist_{>0}}}
\]
and
\[
Y'_1 = \frac{\Weight{\Dist'}-\WtdSum{\Dist'_{<0}}+\WtdSum{\Dist'_{>0}}}
{\Weight{\Dist'_{>0}}-\WtdSum{\Dist'_{<0}}+\WtdSum{\Dist'_{>0}}}.
\]
Let $Y^*_1$ denote the intermediate expression
\[
\frac{\Weight{\Dist'}-\WtdSum{\Dist'_{<0}}+\WtdSum{\Dist'_{>0}}}
{\Weight{\Dist_{>0}}-\WtdSum{\Dist'_{<0}}+\WtdSum{\Dist'_{>0}}}.
\]
Since $\Weight{\Dist'}=\Weight{\Dist}$, $\Dist'_{<0}=\Dist_{<0}$,
$0\leq \Loc'<\Loc^*$,
$0\leq\WtdSum{\Dist'_{>0}}-\Wt^*\Loc'=\WtdSum{\Dist_{>0}}-\Wt^*\Loc^*$,
$\Wt^*>0$, and
$\Weight{\Dist}-\WtdSum{\Dist_{<0}}>\Weight{\Dist_{>0}}-\WtdSum{\Dist_{<0}}>0$,
Fact~\ref{fact:fiveThirdsMono} implies that $Y^*_1>Y_1$. Since
$\Weight{\Dist'_{>0}}\leq\Weight{\Dist_{>0}}$, we have $Y_1'\geq
Y^*_1$.  Hence $Y_1'>Y_1$.

Claim~5: $Z'_{-1}>Z_{-1}$. Observe that
$\Weight{\Dist_{<0}}+\WtdSum{\Dist_{<0}}\geq\Wt^-+\Wt^-\Loc^-$ and
$\Weight{\Dist_{>0}}\geq\WtdSum{\Dist_{>0}}>-\WtdSum{\Dist_{<0}}$. Thus
we obtain the inequality
$\Weight{\Dist}-\WtdSum{\Dist_{<0}}-\Wt^-\Loc^->\Wt^--\WtdSum{\Dist_{<0}}$,
which will be used below.  Lemma~\ref{lem:fiveThirdsEval}
implies that
\begin{eqnarray*}
Z_{-1}
& = &
\frac{\Weight{\Dist}+\WtdSum{\Dist}-2\Wt^-\Loc^-}
{\Wt^--\WtdSum{\Dist_{<0}}+\WtdSum{\Dist_{>0}}}\\
& = &
\frac{\Weight{\Dist}-\WtdSum{\Dist_{<0}}-2\Wt^-\Loc^-+\WtdSum{\Dist_{>0}}}
{\Wt^--\WtdSum{\Dist_{<0}}+\WtdSum{\Dist_{>0}}}
\end{eqnarray*}
and
\[
Z'_{-1} = \frac{\Weight{\Dist'}-\WtdSum{\Dist'_{<0}}-2\Wt^-\Loc^-+\WtdSum{\Dist'_{>0}}}
{\Wt^--\WtdSum{\Dist'_{<0}}+\WtdSum{\Dist'_{>0}}}.
\]
Since $\Weight{\Dist'}=\Weight{\Dist}$, $\Dist'_{<0}=\Dist_{<0}$,
$0\leq \Loc'<\Loc^*$,
$0\leq\WtdSum{\Dist'_{>0}}-\Wt^*\Loc'=\WtdSum{\Dist_{>0}}-\Wt^*\Loc^*$,
$\Wt^*>0$, and $\Weight{\Dist}-\WtdSum{\Dist_{<0}}-2\Wt^-\Loc^->
\Weight{\Dist}-\WtdSum{\Dist_{<0}}-\Wt^-\Loc^->\Wt^--\WtdSum{\Dist_{<0}}>0$,
Fact~\ref{fact:fiveThirdsMono} implies that Claim~5 holds.

Claim~6: If $|\Dist_{>0}|=2$ then $Z'_1>Z_1$. Assume that
$|\Dist_{>0}|=2$.  Lemma~\ref{lem:fiveThirdsEval} implies that
$Z_1=A/B$ where $A=\Weight{\Dist}-\WtdSum{\Dist}+2\Wt^+\Loc^+>0$ and
$B=\Wt^+-\WtdSum{\Dist_{<0}}+\WtdSum{\Dist_{>0}}>0$.  Similarly,
Lemma~\ref{lem:fiveThirdsEval} implies that $Z_1'=A'/B'$ where
$A'=\Weight{\Dist'}-\WtdSum{\Dist'}+2\Wt^+\Loc^+>0$ and
$B'=\Wt^+-\WtdSum{\Dist'_{<0}}+\WtdSum{\Dist'_{>0}}>0$.  Since
$\Weight{\Dist'}=\Weight{\Dist}$ and $\WtdSum{\Dist'}<\WtdSum{\Dist}$,
we have $A'>A$.  Since $\Dist'_{<0}=\Dist_{<0}$ and
$\WtdSum{\Dist'_{>0}}<\WtdSum{\Dist_{>0}}$, we have $B'<B$. It follows
that $Z'_1>Z_1$.
\end{proof}

When the agent locations are described by a special distribution
$\Dist$ such that $\WtdSum{\Dist}\not=0$,
Lemma~\ref{lem:fiveThirdsBalance} below enables us to bound the
performance of Mechanism~\ref{mech:sum-min-fivethird-sgsp-interval} in
terms of its performance on a special distribution $\Dist'$ such that
$\WtdSum{\Dist'}=0$.

\begin{lemma}
\label{lem:fiveThirdsBalance}
Let $\Dist$ be a special distribution such that
$\WtdSum{\Dist}\not=0$. Then there exists a special distribution
$\Dist'$ such that
\[
\min_{\LocFac\in\{-1,0,1\}}\RatTwo{\Dist'}{\LocFac}
>\min_{\LocFac\in\{-1,0,1\}}\RatTwo{\Dist}{\LocFac},
\]
and $\WtdSum{\Dist'}=0$.
\end{lemma}
\begin{proof}
  Let $\Dist$ be a special distribution such that
  $\WtdSum{\Dist}\not=0$.  If $\WtdSum{\Dist}>0$, we can establish the
  claim using at most two applications of
  Lemma~\ref{lem:fiveThirdsBalanceHelper}.  If $\WtdSum{\Dist}<0$,
  we proceed in the same fashion by appealing to a symmetric version
  of Lemma~\ref{lem:fiveThirdsBalanceHelper}.
\end{proof}

When the agent locations are described by a special distribution
$\Dist$ such that $\WtdSum{\Dist}=0$, and locating the facility in
$\{-1,1\}$ does not guarantee an approximation ratio of at most $5/3$,
Lemma~\ref{lem:fiveThirdsRatioPlusMinus} below provides a useful lower
bound on $\WtdSum{\Dist_{>0}}$.

\begin{lemma}
\label{lem:fiveThirdsRatioPlusMinus}
Let $\Dist$ be a special distribution such that $\WtdSum{\Dist}=0$ and
\[
\min_{\LocFac\in\{-1,1\}}\RatTwo{\Dist}{\LocFac}>\frac{5}{3}.
\]
Then $\WtdSum{\Dist_{>0}}>\Weight{\Dist}/3$.
\end{lemma}
\begin{proof}
Since $\WtdSum{\Dist}=0$, we have
$\WtdSum{\Dist_{>0}}=-\WtdSum{\Dist_{<0}}$.  Thus
Lemma~\ref{lem:fiveThirdsPlusMinusB} implies that
\[
\min_{\LocFac\in\{-1,1\}}\RatTwo{\Dist}{\LocFac}=
1+\frac{2\WtdSum{\Dist_{>0}}}{\Weight{\Dist}}.
\]
The claim of the lemma follows.  
\end{proof}

When the agent locations are described by a special distribution
$\Dist$ such that $\WtdSum{\Dist}=0$, and the inequality
$\WtdSum{\Dist_{>0}}>\Weight{\Dist}/3$ appearing in the statement of
Lemma~\ref{lem:fiveThirdsRatioPlusMinus} is satisfied,
Lemma~\ref{lem:fiveThirdsRatioZeroHelper} below shows that locating
the facility at $0$ results in performance within a factor of $5/3$ of
that obtained by locating it at $1$. A symmetric claim shows that
locating the facility at $0$ results in performance within a factor of
$5/3$ of that obtained by locating it at $-1$.

\begin{lemma}
\label{lem:fiveThirdsRatioZeroHelper}
  Let $\Dist$ be a special distribution such that $\WtdSum{\Dist}=0$
  and $\WtdSum{\Dist_{>0}}>\Weight{\Dist}/3$. Then
\[
\max_{k\in\{0,\ldots,|\Dist_{>0}|\}}
\ObjRat{\Dist}{\Suffix{\Dist}{k}}{0}{1}
<
\frac{5}{3}.
\]
\end{lemma}
\begin{proof}
Let $X$ denote $\ObjRat{\Dist}{\emptyset}{0}{1}$, let $Y$ denote
$\ObjRat{\Dist}{\Dist_{>0}}{0}{1}$, let $(\Loc,\Wt)$ denote the
lexicographically greatest pair in $\Dist$, and let $Z$ denote
$\ObjRat{\Dist}{\{(\Loc,\Wt)\}}{0}{1}$.  It is sufficient to prove
that $\max(X,Y,Z)<\frac{5}{3}$. The latter inequality is immediate
from Claims~1 through~3 below. In the proofs of Claims~1 through~3,
let $\Delta$ denote $\WtdSum{\Dist_{>0}}$, which is equal to
$-\WtdSum{\Dist_{<0}}$ since $\WtdSum{\Dist}=0$.

Claim~1: $X<\frac{3}{2}$. We have $X=\frac{\Weight{\Dist}}{2\Delta}$,
and the claimed inequality follows since $\Delta>\Weight{\Dist}/3$.

Claim~2: $Y<\frac{5}{3}$. We have
$Y=\frac{\Weight{\Dist}+2\Delta}{\Weight{\Dist_{>0}}+2\Delta}$. Since
$\Weight{\Dist_{>0}}\geq\Delta$, we have
$Y\leq\frac{\Weight{\Dist}+2\Delta}{3\Delta}$. Since the latter bound
is strictly decreasing in $\Delta$ and $\Delta>\Weight{\Dist}/3$, the
desired inequality follows.

Claim~3: $Z<\frac{5}{3}$. We have
$Z=\frac{\Weight{\Dist}+2\Wt\Loc}{\Wt+2\Delta}$.  Since $\Wt>0$ and
$\Wt\leq\Delta/\Loc$, Fact~\ref{fact:fiveThirdsMono} (viewing $Z$
as a function of $\Wt$) implies that $Z$ is at most
\[
\max\left(\frac{\Weight{\Dist}}{2\Delta},\frac{\Weight{\Dist}+2\Delta}{(2+\frac{1}{\Loc})\Delta}\right).
\]
Since $\Delta>\Weight{\Dist}/3$, the first argument above is less than
$3/2$. Since $0<\Loc<1$, the second argument above is less than
$\frac{\Weight{\Dist}+2\Delta}{3\Delta}$, which is less than $5/3$.
\end{proof}

When the agent locations are described by a special distribution
$\Dist$ such that $\WtdSum{\Dist}=0$, and the inequality
$\WtdSum{\Dist_{>0}}>\Weight{\Dist}/3$ appearing in the statement of
Lemma~\ref{lem:fiveThirdsRatioPlusMinus} is satisfied,
Lemma~\ref{lem:fiveThirdsRatioZero} below shows that locating
the facility at $0$ guarantees an approximation ratio less than $5/3$.

\begin{lemma}
\label{lem:fiveThirdsRatioZero}
  Let $\Dist$ be a special distribution such that $\WtdSum{\Dist}=0$
  and $\WtdSum{\Dist_{>0}}>\Weight{\Dist}/3$. Then
  $\RatTwo{\Dist}{0}<\frac{5}{3}$.
\end{lemma}
\begin{proof}
Using an argument that is symmetric to the proof of
Lemma~\ref{lem:fiveThirdsRatioZeroHelper}, we obtain
\[
\max_{k\in\{0,\ldots,|\Dist_{<0}|\}}
\ObjRat{\Dist}{\Prefix{\Dist}{k}}{0}{-1}
<
\frac{5}{3}.
\]
Combining the above inequality with
Lemma~\ref{lem:fiveThirdsRatioZeroHelper}, the claim of the lemma
follows by Lemma~\ref{lem:fiveThirdsPrefixSuffix}.
\end{proof}

We are now ready to bound the approximation ratio of
Mechanism~\ref{mech:sum-min-fivethird-sgsp-interval}.

\begin{theorem}
\label{thm:fiveThirds}
Mechanism~\ref{mech:sum-min-fivethird-sgsp-interval} achieves an approximation ratio of $5/3$.
\end{theorem}
\begin{proof}
Let $\Dist$ be a distribution. Given the definition of Mechanism~\ref{mech:sum-min-fivethird-sgsp-interval},
we need to prove that
\[
\min_{\LocFac\in\{-1,0,1\}}\RatTwo{\Dist}{\LocFac}\leq\frac{5}{3}.
\]
If $\Dist$ is trivial, it is straightforward to argue that
$\min_{\LocFac\in\{-1,0,1\}}\RatTwo{\Dist}{\LocFac}=1$.  If $\Dist$ is
nontrivial, the desired inequality follows from
Lemmas~\ref{lem:fiveThirdsNontrivial}, \ref{lem:fiveThirdsBalance},
\ref{lem:fiveThirdsRatioPlusMinus}, and~\ref{lem:fiveThirdsRatioZero}.
\end{proof}

\subsection{The cycle}
\label{subsec:sum-min-cycle}

Now we present a simple adaptation of
Mechanism~\ref{mech:sum-min-sgsp} for the case where the agents are
located on a cycle.

\begin{custommech}{5}
\label{mech:sum-min-sgsp-circle}
Let $\Instance$ denote the reported $\DOFL$ instance. Build all of the
facilities at $0$ if
\begin{equation*}
    \sum_{i\in [n]} \dist(x_i,0) \ge \sum_{i \in [n]} \dist(x_i, 1/2);
\end{equation*}
otherwise, build all of the facilities at $1/2$.
\end{custommech}

As with Mechanism~\ref{mech:sum-min-sgsp}, reported dislikes do not affect the locations at which Mechanism~\ref{mech:sum-min-sgsp-circle} builds the facilities. Hence Mechanism~\ref{mech:sum-min-sgsp-circle} is SGSP.

\begin{theorem}
\label{thm:sum-min-sgsp-circle-sgsp}
Mechanism~\ref{mech:sum-min-sgsp-circle} is SGSP.
\end{theorem}

\begin{theorem}
\label{thm:sum-min-sgsp-circle-opt}
Mechanism~\ref{mech:sum-min-sgsp-circle} is $2$-efficient.
\end{theorem}
\begin{proof}
We sketch a proof that is similar to our proof of
Theorem~\ref{thm:opt-sum-min-sgsp}.  Let $I = \Instance$ denote the
reported $\DOFL$ instance. Let $\ALG$ denote the social welfare
obtained by Mechanism~\ref{mech:sum-min-sgsp-circle} on this instance,
and let $\OPT$ denote the maximum possible social welfare on this
instance.  We need to prove that $2\cdot\ALG\geq\OPT$.

Assume without loss of generality that
Mechanism~\ref{mech:sum-min-sgsp-circle} builds all of the facilities
at $0$. (A symmetric argument handles the case where all of the
facilities are built at $1/2$.)  Using similar arguments as in the
proof of Theorem~\ref{thm:opt-sum-min-sgsp}, we obtain $\ALG \ge
\sum_{i\in [n]} \dist(x_i,0)$. As
Mechanism~\ref{mech:sum-min-sgsp-circle} builds the facilities at $0$
and not $1/2$, we have $\sum_{i\in [n]} \dist(x_i,0) \ge \sum_{i \in
  [n]} \dist(x_i, 1/2)$. We also have $\dist(x_i, 0) + \dist(x_i, 1/2)
= 1/2$ for all agents $i$. Thus $\sum_{i\in [n]} \dist(x_i, 0) \ge
n/4$. Hence $\ALG \ge n/4$.  Since no agent has welfare greater than
$1/2$, we have $n/2 \ge \OPT$. Thus $2 \cdot \ALG \ge n/2 \ge \OPT$,
as required.
\end{proof}

\subsection{The unit square}
\label{subsec:sum-min-square}

We now show how to adapt Mechanism~\ref{mech:sum-min-sgsp} to the case where the agents are located in the unit square.

\begin{custommech}{6}
\label{mech:sum-min-sgsp-square}
Let $\Instance$ denote the reported $\DOFL$ instance.  For each point
$p$ in $\{0,1\}^2$, let $d_p$ denote $\sum_{i\in [n]} \dist(x_i,p)$.
Let $q$ be the point in $\{0,1\}^2$ that maximizes $d_q$, breaking
ties lexicographically.  Build all of the facilities at $q$.
\end{custommech}

Mechanism~\ref{mech:sum-min-sgsp-square} computes $4n$ Euclidean
  distances, and runs in $O(n)$ time.  As in the case of
Mechanism~\ref{mech:sum-min-sgsp}, reported dislikes do not affect the
locations at which Mechanism~\ref{mech:sum-min-sgsp-square} builds the
facilities. Hence Mechanism~\ref{mech:sum-min-sgsp-square} is SGSP.

\begin{theorem}
\label{thm:sum-min-sgsp-square-sgsp}
Mechanism~\ref{mech:sum-min-sgsp-square} is SGSP.
\end{theorem}

\begin{theorem}
\label{thm:sp-sum-min-sgsp-square-opt}
Mechanism~\ref{mech:sum-min-sgsp-square} is $2$-efficient.
\end{theorem}
\begin{proof}
We sketch a proof that is similar to our proof of
Theorem~\ref{thm:opt-sum-min-sgsp}.  Let $I = \Instance$ denote the
reported $\DOFL$ instance. Let $\ALG$ denote the social welfare
obtained by Mechanism~\ref{mech:sum-min-sgsp-square} on this instance,
and let $\OPT$ denote the maximum possible social welfare on this
instance.  We need to prove that $2\cdot\ALG\geq\OPT$.

Assume without loss of generality that
Mechanism~\ref{mech:sum-min-sgsp-square} builds all of the facilities
at $(0, 0)$. (A symmetric argument handles the three remaining cases.)
Using similar arguments as in the proof of
Theorem~\ref{thm:opt-sum-min-sgsp}, we obtain $\ALG \ge \sum_{i\in
  [n]} \dist(x_i,(0,0))$.  As Mechanism~\ref{mech:sum-min-sgsp-square}
builds the facilities at $(0,0)$, we have $$\sum_{i\in [n]}
\dist(x_i,(0,0)) \ge \max_{p\in \{(0,1), (1,0), (1,1)\}} \sum_{i \in
  [n]} \dist(x_i, p).$$ We also have $\dist(x_i, (0,0)) + \dist(x_i,
(0,1)) + \dist(x_i, (1,0)) + \dist(x_i, (1,1))\ge 2\sqrt{2}$ for all
agents $i$. Thus $\sum_{i\in [n]} \dist(x_i, (0,0)) \ge
n/\sqrt{2}$. Hence $\ALG \ge n/\sqrt{2}$.  Since no agent has welfare
greater than $\sqrt{2}$, we have $\sqrt{2}n \ge \OPT$.  Thus $2 \cdot
\ALG \ge \sqrt{2}n \ge \OPT$, as required.
\end{proof}

\section{Egalitarian Mechanisms}
\label{sec:min-min}

We now design egalitarian mechanisms for $\OOFLG$ when the agents are
located on an interval, cycle, or square.

In Definition~\ref{def:parallel} below, we introduce a simple way to
convert a single-facility $\OOFLG$ mechanism into a $\OOFLG$
mechanism.  For any $\DOFL$ instance $I = \Instance$ and any $j$
  in $[k]$, let $\single(I,j)$ denote the single-facility $\DOFL$
  instance $(n, 1, \Location, \Aversion')$ where $a'_i$ is the
  singleton set containing the sole facility if $i$ belongs to
  $\haters(I,j)$, and is $\emptyset$ otherwise.

\begin{definition}
\label{def:parallel}
For any single-facility $\OOFLG$ mechanism $M$, let $\Parallel{M}$ denote
the $\OOFLG$ mechanism that takes as input a $\DOFL$ instance
$I=\Instance$ and outputs $\mathbf{y} = (y_1, \dots, y_k)$ where
$y_j$ is the location at which $M$ builds the facility on input $\single(I,j)$.
\end{definition}

Lemmas~\ref{lem:min-min-sfo-ooflg-sp} and
\ref{lem:min-min-sfo-ooflg-opt} below reduce the task of designing
a SP egalitarian $\OOFLG$ mechanism to the single-facility case. 

\begin{restatable}{lemma}{minminsfoooflgsp}
\label{lem:min-min-sfo-ooflg-sp}
Let $M$ be a SP single-facility $\OOFLG$ mechanism. Then
$\Parallel{M}$ is a SP $\OOFLG$ mechanism.
\end{restatable}
\begin{proof}
Let $(I,I')$ denote a $\OOFLG$ instance with $I=\Instance$ and
$I'=\LieInstance$, and let $i$ be an agent such that $\LieAversion =
(\Aversion_{-i}, a_i')$.  Let $\mathbf{y} = (y_1, \dots, y_k)$ (resp.,
$\mathbf{y'} = (y'_1, \dots, y'_k)$) denote $\Parallel{M}(I)$ (resp.,
$\Parallel{M}(I')$).  Since $M$ is SP, we have $\dist(x_i, y_j) \ge
\dist(x_i, y_j')$ for each facility $F_j$ in $a_i$. Thus $w(I,i,
\mathbf{y}) \ge w(I, i, \mathbf{y'})$, implying that agent $i$ does
not benefit by reporting $a_i'$ instead of $a_i$.
\end{proof}

\begin{restatable}{lemma}{minminsfoooflgopt}
\label{lem:min-min-sfo-ooflg-opt}
Let $M$ be an egalitarian single-facility $\OOFLG$ mechanism. Then
$\Parallel{M}$ is an egalitarian $\OOFLG$ mechanism.
\end{restatable}
\begin{proof}
Let $I=\Instance$ denote the reported $\DOFL$ instance.  Let an
optimal solution be $\mathbf{y^*} = (y_1^*, \dots, y_k^*)$, and let
the optimal (maximum) value of the minimum welfare be $\OPT = \MMW(I,
\mathbf{y^*})$. Assume that $\Parallel{M}$ builds the facilities at
$\mathbf{y}' = (y'_1, \dots, y'_k)$, resulting in minimum welfare
$\ALG = \MMW(I, \mathbf{y}')$.  For each facility $F_j$, we define
$\OPT_j$ (resp., $\ALG_j$) as the distance from $y_j^*$ (resp.,
$y'_j$) to the nearest agent in $\haters(I, j)$ (or $\infty$ if
$\haters(I, j)$ is empty).

We have
\begin{equation*}
    \OPT = \min\left(\min_j \OPT_j, \min_{i\in \indiff(I)} w(I,i, \mathbf{y^*})\right)
\end{equation*}
and 
\begin{equation*}
    \ALG = \min\left(\min_j \ALG_j, \min_{i\in \indiff(I)} w(I, i, \mathbf{y}')\right).
\end{equation*}
Since $M$ is egalitarian, we have $\OPT_j = \ALG_j$ for all $j$. The
welfare of the agents in $\indiff(I)$ does not depend on the locations
of the facilities. Thus $\ALG = \OPT$, implying that $\Parallel{M}$ is
egalitarian.
\end{proof}

\subsection{The unit interval}
\label{subsec:min-min-interval}

We begin by describing a SP egalitarian mechanism for single-facility
$\OOFLG$ when the agents are located in the unit interval.

\begin{custommech}{7}
\label{mech:min-min}
Let $I = (n, 1, \Location, \Aversion)$ denote the reported $\DOFL$
instance and let $H$ denote $\haters(I, 1)$.  If $H$ is empty, build
$F_1$ at $0$. Otherwise, let $H$ contain $\ell$ agents $z_1, \dots,
z_\ell$ such that $x_{z_1} \le x_{z_2} \le \dots \le x_{z_\ell}$. Let
$d_1 = x_{z_1}$ and $d_3 = 1 - x_{z_\ell}$. If $\ell = 1$, then build
$F_1$ at $0$ if $d_1 \ge d_3$, and at $1$ otherwise.  If $\ell>1$, let
$m$ be the midpoint of the leftmost largest interval between
consecutive agents in $H$. Formally, $m = (x_{z_s} + x_{z_{s+1}}) /
2$, where $s$ is the smallest number in $[\ell-1]$ such that
$x_{z_{s+1}} - x_{z_s} = \max_{j \in [\ell-1]} (x_{z_{j+1}} -
x_{z_j})$. Let $d_2 = m - x_{z_s}$. Then build facility $F_1$ at $0$
if $d_1 \ge d_2$ and $d_1 \ge d_3$, at $m$ if $d_2 \ge d_3$, and at
$1$ otherwise.
\end{custommech}

The following lemma shows that Mechanism~\ref{mech:min-min} is SP. It
is established by examining how the location of the facility changes
when an agent misreports.

\begin{restatable}{lemma}{minminspsfo}
\label{lem:min-min-sp-sfo}
Mechanism~\ref{mech:min-min} is SP for single-facility $\OOFLG$.
\end{restatable}
\begin{proof}
Let $(I,I')$ denote a single-facility $\OOFLG$ instance with $I = (n,
1, \Location, \Aversion)$ and $I' = (n, 1, \Location, \LieAversion)$, and
let $i$ be an agent such that $\LieAversion = (\Aversion_{-i}, a_i')$.
If $F_1$ does not belong to $a_i$, the welfare of agent $i$ is
independent of the location of $F_1$ and agent $i$ does not benefit by
reporting $a'_i$. Moreover, if $F_1$ belongs to $a_i \cap a_i'$, then
the location of $F_1$ does not change by reporting $a'_i$ instead of
$a_i$, and again, agent $i$ does not benefit by reporting $a'_i$. Accordingly,
for the remainder of the proof, we assume that $F_1$ belongs to $a_i
\setminus a'_i$.

Let $y$ denote the location at which Mechanism~\ref{mech:min-min}
builds $F_1$ when agent $i$ reports truthfully, and let $H$ denote
$\haters(I, 1)$.  Note that Mechanism~\ref{mech:min-min} does not
build $F_1$ at the location of any agent in $H$, that is, $y \neq
x_{i'}$ for all $i'$ in $H$.  Hence $y \neq x_i$.  We can assume
without loss of generality that $y < x_i$, since the case $y > x_i$
can be handled symmetrically.  Let $d_1$, $d_2$, and $d_3$ be as
defined in Mechanism~\ref{mech:min-min} when all agents report
truthfully.  We consider two cases based on whether there is an agent
in $H$ between $y$ and $x_i$.

Case~1: No agent in $H-i$ is located in $[y,x_i]$. We consider two
cases.

Case~1.1: $y = 0$. Thus $d_1 = x_i$. When agent $i$ reports $a'_i$, 
$F_1$ is built at $0$, which does not benefit agent $i$. 

Case~1.2: $y \neq 0$. Thus $d_2 = x_i - y$, there is an agent $i'$ in
$H$ at $y - d_2$, and there are no agents in $H$ in $(y - d_2, y +
d_2)$. We consider two cases.

Case~1.2.1: No agent in $H$ is located to the right of $x_i$. Hence
$x_i \ge 1 - d_2$. Thus when agent $i$ reports $a'_i$, $F_1$ is built
at $1$, which does not benefit agent $i$.

Case~1.2.2: There is an agent in $H$ located to the right of
$x_i$. Let $i''$ be the first agent to the right of $x_i$, breaking
ties arbitrarily. Then $x_{i''} - x_i \le 2d_2$. Thus when agent $i$
reports $a'_i$, $F_1$ is built in $[y, x_i]$, which does not benefit
agent $i$.

Case~2: There is an agent in $H - i$ in $[y, x_i]$. Let $i'$ be the
first agent to the right of $y$ in $H - i$.  Let $d$ denote $d_1 =
x_{i'}$ if $y = 0$, and $d_2 = x_{i'} - y$ otherwise. It follows that
$x_i - y \ge d$. We consider two cases.

Case~2.1: No agent in $H$ is located to the right of $x_i$. Hence $x_i
\ge 1 - d$. Thus when agent $i$ reports $a'_i$, $F_1$ is either built
at $y$ or at $1$, neither of which benefits agent $i$.

Case~2.2: There is an agent in $H$ located to the right of $x_i$. Let
$b$ be the first agent to the right of $x_i$, breaking ties
arbitrarily.  Let agent $a$ be the agent located in $[0, x_i]$ that is
closest to agent $i$, breaking ties arbitrarily.  It follows that $x_i
- x_a \le 2d$ and $x_b - x_i \le 2d$. When agent $i$ reports $a'_i$,
$F_1$ is built at $y$ or in $[x_i - d, x_i + d]$, neither of which
benefits agent $i$.

Thus agent $i$ does not benefit by reporting $a'_i$.
\end{proof}

\begin{lemma}
\label{lem:min-min-opt-sfo}
Mechanism~\ref{mech:min-min} is egalitarian for single-facility $\OOFLG$.
\end{lemma}
\begin{proof}
Let $I = (n, 1, \Location, \Aversion)$ denote the reported $\DOFL$
instance, let $H$ denote $\haters(I, 1)$,
let $y^*$ denote an optimal location for the facility, let $\OPT$
denote $\MMW(I,y^*)$, let $y'$ denote the location at which
Mechanism~\ref{mech:min-min} builds the facility, and let $\ALG$
denote $\MMW(I,y')$. Below we establish that $\ALG \ge \OPT$.  Since
$\ALG \le \OPT$, we conclude that $\ALG = \OPT$ and hence that
Mechanism~\ref{mech:min-min} is egalitarian.

If $H$ is empty, then trivially Mechanism~\ref{mech:min-min} is
egalitarian. For the remainder of the proof, assume that $H$ is
non-empty. We say that an agent in $H$ is \emph{tight} if it is as
close to $y^*$ as any other agent in $H$.  Thus for any tight agent
$i$, $\OPT = |y^* - x_i|$. Similarly, $\ALG$ is the distance between
$y'$ and a closest agent in $H$. Let $i$ be a tight agent, and
consider the following three cases.

Case~1: $y^*=0$.  In this case, no agent in $H$ is located in $[0,
  x_i)$. It follows that $d_1 = x_i = \OPT$. Since $\ALG\ge d_1$, we
  have $\ALG \ge \OPT$.

Case~2: $y^*=1$. Symmetric to Case~1.

Case~3: $0<y^*<1$.  Since $y^* = x_i$ implies $\OPT = 0$, it is easy
to see that $y^* \neq x_i$.  We can assume without loss of generality
that $x_i < y^*$, since the case $x_i > y^*$ can be handled
symmetrically.  Thus $\OPT = y^* - x_i$ and no agent in $H$ is located
in $(x_i = y^* - \OPT, y^* + \OPT)$. We consider two cases.

Case~3.1: There is no agent in $H$ to the right of $y^*$. Thus $d_3
\ge \OPT$. Since $\ALG \ge d_3$, we have $\ALG \ge \OPT$.

Case~3.2: There is an agent in $H$ to the right of $y^*$. Consider the
leftmost such agent $i'$. Since $x_{i'} \ge y^* + \OPT$, we have $d_2
\ge \OPT$. Since $\ALG \ge d_2$, we have $\ALG \ge \OPT$.
\end{proof}

We define Mechanism~8 as the $\OOFLG$ mechanism $\Parallel{M}$, where
$M$ denotes Mechanism~\ref{mech:min-min}.  Using
Lemmas~\ref{lem:min-min-sfo-ooflg-sp} through
\ref{lem:min-min-opt-sfo}, we immediately obtain
Theorem~\ref{thm:min-min-opt-eg} below.

\begin{theorem}
\label{thm:min-min-opt-eg}
Mechanism~8 is SP and egalitarian.
\end{theorem}

Below we provide a lower bound on the approximation ratio of any WGSP
egalitarian mechanism.
Theorem~\ref{thm:min-min-lowerbound} implies that Mechanism~8 is not WGSP.

\begin{theorem}
\label{thm:min-min-lowerbound}
Let $M$ be a WGSP $\alpha$-egalitarian mechanism. Then $\alpha$ is
$\Omega(\sqrt{n})$, where $n$ denotes the number of agents.
\end{theorem}
\begin{proof}
Let $q$ be a large even integer, let $p$ denote $q^2 + 1$, and let $U$
(resp., $V$) denote the set of all integers $i$ such that $0<i<q^2/2$
(resp., $q^2/2<i<q^2$).  We construct two $(p+3)$-agent two-facility
$\OOFLG$ instances $(I,I)$ and $(I,I')$.
In both $(I,I)$ and $(I,I')$, there is an agent located at $i/q^2$,
called agent $i$, for each $i$ in $U \cup V$, and there are two agents
each at $0$, $1/2$, and $1$.  In $I$, each agent $i$ in $U$ dislikes
$\{F_2\}$, each agent $i$ in $V$ dislikes $\{F_1\}$, one agent at $0$
(resp., $1/2$, $1$) dislikes $\{F_1\}$, and the other agent at $0$
(resp., $1/2$, $1$) dislikes $\{F_2\}$.
In $I'$, the agents $i$ in $U \setminus \{1, \dots, q-1\}$ have
alternating reports: agent $q$ reports $\{F_1\}$, agent $q+1$ reports
$\{F_2\}$, agent $q+2$ reports $\{F_1\}$, and so on. Symmetrically,
the agents $i$ in $V \setminus \{q^2-q+1, \dots, q^2-1\}$ have
alternating reports: agent $q^2-q$ reports $\{F_2\}$, agent $q^2-q-1$
reports $\{F_1\}$, agent $q^2-q-2$ reports $\{F_2\}$, and so on. All
other agents in $I'$ report truthfully.


Let the optimal minimum welfare for $\DOFL$ instance $I$ (resp., $I'$)
be $\OPT$ (resp., $\OPT'$).  It is easy to see that $\OPT=1/4$ and
$\OPT'=\frac{1}{2q}$ (obtained by building the facilities at $(1/4,
3/4)$ and $(\frac{1}{2q}, 1-\frac{1}{2q})$, respectively).  Let $\ALG$
(resp., $\ALG'$) denote the minimum welfare achieved by $M$ on
instance $I$ (resp., $I'$).  Below we prove that either
$\OPT/\ALG\ge q/4$ or $\OPT'/\ALG'\ge q/2$.


Let $M$ build facilities at $(y_1, y_2)$ (resp., $(y'_1, y'_2)$) on
instance $I$ (resp., $I'$). We consider two cases.

Case~1: $0 \le y'_1 < 1/q$ and $1-1/q < y'_2 \le 1$.  Let $S$ denote
the set of agents who lie in $I'$.
If $y'_1 < y_1$ and $y'_2 > y_2$, then all of the agents in $S$
benefit by lying. Hence for $M$ to be WGSP, either $y'_1 \ge y_1$ or
$y'_2 \le y_2$. Let us assume that $y'_1 \ge y_1$; the case where
$y'_2\le y_2$ can be handled symmetrically. Since $y'_1 < 1/q$, we
have $y_1 < 1/q$. Note that there is an agent at $0$ who reports
$\{F_1\}$. Thus $\ALG \le y_1 < 1/q$. Hence $\OPT/\ALG\ge
q/4$.

Case~2: $y'_1 \ge 1/q$ or $y'_2 \le 1-1/q$. If $y'_1 \ge 1/q$, then at least one agent within distance $1/q^2$ of
$y'_1$ reported $\{F_1\}$ in $I'$.  A similar observation holds for the case
$y'_2 \le 1-1/q$. Thus $\ALG' \le 1/q^2$. Hence $\OPT' / \ALG' \ge
q/2$.

The preceding case analysis shows that $\alpha \ge q/4$.  Since $q =
\sqrt{p-1} = \sqrt{n-4}$, the theorem holds.
\end{proof}

The following variant of Mechanism~8 is SGSP. In this variant, we
treat the reported dislikes of all agents as if they were $\{F_1\}$,
and we use Mechanism~\ref{mech:min-min} to determine where to build
$F_1$. Then we build all of the remaining facilities at the same
location as $F_1$. This mechanism is SGSP because it disregards the
reported aversion profile. We claim that this mechanism is
$2(n+1)$-egalitarian, where $n$ denotes the number of agents. To prove
this claim, we first observe that when Mechanism~\ref{mech:min-min} is
run as a subroutine within this mechanism, we have $\max(d_1, 2d_2,
d_3) \ge 1/(n+1)$. Thus the minimum welfare achieved by the mechanism
is at least $\frac{1}{2(n+1)}$.  Since the optimal minimum welfare is at
most $1$, the claim holds.

\subsection{The cycle}
\label{subsec:min-min-cycle}

In this section, we present egalitarian mechanisms for the case where
the agents are located on the unit-circumference circle $C$. For any
point $u$ on $C$, we let $\widehat{u}$ denote the point antipodal to
$u$. We begin by considering the natural adaptation of
Mechanism~\ref{mech:min-min} to a cycle.

\begin{custommech}{9}
\label{mech:min-min-circle}
Let $I = (n, 1, \Location, \Aversion)$ denote the reported $\DOFL$
instance and let $H$ denote $\haters(I, 1)$.  If $H$ is empty, then
build facility $F_1$ at $0$. If $H$ has only one agent $i$, then build
$F_1$ at $\widehat{x_i}$.  Otherwise, build $F_1$ at the midpoint of
the largest gap between any two consecutive agents in $H$. Formally,
let $H$ have $\ell$ agents $z_0, \dots, z_{\ell-1}$ such that $x_{z_0}
\le x_{z_1} \le \dots \le x_{z_{\ell-1}}$. Let $\oplus$ denote
addition modulo $\ell$. Build $F_1$ at the midpoint of $x_{z_s}$ and
$x_{z_{s\oplus 1}}$, where $s$ is the smallest number in $\{0, \dots
\ell-1\}$ such that $\dist(x_{z_{s\oplus 1}}, x_{z_s}) = \max_{j \in
  \{0, \dots \ell-1\} } \dist(x_{z_{j\oplus 1}}, x_{z_j})$.
\end{custommech}

\begin{lemma}
\label{lem:min-min-circle-sp}
Mechanism~\ref{mech:min-min-circle} is SP for single-facility
$\OOFLG$.
\end{lemma}
\begin{proof}
Let $(I,I')$ denote a single-facility $\OOFLG$ instance with $I = (n,
1, \Location, \Aversion)$ and $I' = (n, 1, \Location, \LieAversion)$,
and let $i$ be an agent such that $\LieAversion = (\Aversion_{-i},
a_i')$.  As in the proof of Lemma~\ref{lem:min-min-sp-sfo}, we can
restrict our attention to the case where $F_1$ belongs to $a_i
\setminus a'_i$.

Let $y$ denote the location at which
Mechanism~\ref{mech:min-min-circle} builds $F_1$ when agent $i$
reports truthfully, and let $H$ denote $\haters(I, 1)$.  Note that
Mechanism~\ref{mech:min-min-circle} does not build $F_1$ at the
location of any agent in $H$, that is, $y \neq x_{i'}$ for all $i'$ in
$H$.  Hence $y \neq x_i$.  Let the arc of $C$ that goes clockwise from
$y$ to $x_i$ be $r_1$ and let the arc of $C$ that goes
counterclockwise from $y$ to $x_i$ be $r_2$. Both arcs $r_1$ and $r_2$
include the endpoints $y$ and $x_i$. We consider four cases.

Case~1: No agent in $H - i$ is on $r_1$ or $r_2$. Hence $H =
\{i\}$. Thus $y = \widehat{x_i}$, and $\dist(x_i, y) = 1/2$. When
agent $i$ reports $a_i'$, $F_1$ is built at $0$. Since $\dist(x_i, 0)
\le 1/2$, reporting $a_i'$ does not benefit agent $i$.

Case~2: 
No agent in $H-i$ is on $r_{1}$ and there is an agent in
$H-i$ on $r_{2}$.
Let $i'$ be the closest agent to $y$ in $H-i$ on $r_2$. Let $d$ denote
$\dist(y, x_{i'})$. Thus $y$ is the midpoint of the arc that runs
clockwise from $x_{i'}$ to $x_i$. Hence $d = \dist(x_i, y)$. Let $i''$
be the closest agent in $H-i$ in the clockwise direction from
$x_i$. Thus $\dist(x_{i''},x_i) \le 2d$.  Since $F_1$ is built on
$r_1$ when agent $i$ reports $a_i'$ and $\dist(x_{i''}, x_i) \le 2d$,
reporting $a_i'$ does not benefit agent~$i$.

Case~3: 
No agent in $H-i$ is on $r_{2}$ and there is an agent in
$H-i$ on $r_{1}$. Symmetric to Case~2.

Case~4: 
There is an agent in $H-i$ on $r_{1}$ and there is an agent in $H-i$ on $r_{2}$. Let the
closest agent from $y$ in $H - i$ on $r_2$ (resp., $r_1$) be $a$
(resp., $b$), respectively. We have $\dist(x_a, y) = \dist(y,
x_b)$. Let $d$ denote $\dist(x_a, y)$. Note that $\dist(x_i, y) \ge
d$. Let $i'$ (resp., $i''$) be the first agent in $H-i$ encountered in
the counterclockwise (resp., clockwise) direction from $x_i$. We have
$\dist(x_i, x_{i'}) \le 2d$ and $\dist(x_i, x_{i''}) \le 2d$. Thus,
when agent $i$ reports $a_i'$, either $F_1$ is built at $y$ or $F_1$
is built within distance $d$ of $x_i$, neither of which benefits agent
$i$.

Thus agent $i$ does not benefit by reporting $a_i'$.
\end{proof}

\begin{lemma}
\label{lem:min-min-circle-opt}
Mechanism~\ref{mech:min-min-circle} is egalitarian for single-facility
$\OOFLG$.
\end{lemma}
\begin{proof}
Let $I = (n, 1, \Location, \Aversion)$ denote the reported $\DOFL$
instance, let $H$ denote $\haters(I, 1)$, let $y^*$ denote an optimal
location for the facility, let $\OPT$ denote $\MMW (I, y^*)$, let $y'$
denote the location at which Mechanism~\ref{mech:min-min-circle}
builds the facility, and let $\ALG$ denote $\MMW (I, y')$.  Below we
establish that $\ALG \ge \OPT$. Since $\ALG \le \OPT$, we conclude
that $\ALG = \OPT$ and hence that Mechanism~\ref{mech:min-min-circle}
is egalitarian.

If $|H|\leq 1$, it is easy to see that
Mechanism~\ref{mech:min-min-circle} is egalitarian. For the remainder
of the proof, we assume that $|H|\geq 2$.  We say that an agent in $H$
is \emph{tight} if it is as close to $y^*$ as any other agent in $H$.
Thus for any tight agent $i$, $\OPT = \dist(y^*,x_i)$.

Let $i$ be a tight agent. Assume without loss of generality that in
the shorter arc between $x_i$ and $y^*$, $x_i$ lies counterclockwise
from $y^*$. Thus $\OPT = \dist(x_i, y^*)$.  Let $i'$ be the closest
agent in $H$ in the clockwise direction from $y^*$.  The definition of
$i'$ implies that $\dist(x_{i'}, y^*) \ge \OPT$.  Thus the length of
the clockwise arc from $x_i$ to $x_{i'}$ is at least $2 \cdot \OPT$.
Since $i$ and $i'$ are consecutive agents in $H$ and
Mechanism~\ref{mech:min-min-circle} builds the facility at the
midpoint of the largest gap between consecutive agents in $H$, we
deduce that $\ALG \ge \OPT$.
\end{proof}

We define Mechanism~10 as the $\OOFLG$ mechanism $\Parallel{M}$, where
$M$ denotes Mechanism~\ref{mech:min-min-circle}.  Using
Lemmas~\ref{lem:min-min-sfo-ooflg-sp},
\ref{lem:min-min-sfo-ooflg-opt}, \ref{lem:min-min-circle-sp},
and~\ref{lem:min-min-circle-opt}, we immediately obtain
Theorem~\ref{thm:min-min-opt-eg-circle} below.

\begin{theorem}
\label{thm:min-min-opt-eg-circle}
Mechanism~10 is SP and egalitarian.
\end{theorem}

Theorem~\ref{thm:min-min-circle-lowerbound} below extends
Theorem~\ref{thm:min-min-lowerbound} to the case of the
cycle. Theorem~\ref{thm:min-min-circle-lowerbound} implies that
Mechanism~10 is not WGSP.

\begin{theorem}
\label{thm:min-min-circle-lowerbound}
Let $M$ be a WGSP $\alpha$-egalitarian mechanism. Then $\alpha$ is
$\Omega(\sqrt{n})$, where $n$ is the number of agents.
\end{theorem}
\begin{proof}
It is straightforward to verify that the construction used in the
proof of Theorem~\ref{thm:min-min-lowerbound} also works for the cycle
and establishes the same lower bound. (We identify the point $1$ with
the point $0$.)
\end{proof}

The following variant of Mechanism~10 is SGSP. As in the SGSP mechanism
for the case when the agents are located in the unit interval, in this
variant, we treat the reported dislikes of all agents as if they were
$\{F_1\}$,
and we use Mechanism~\ref{mech:min-min-circle} to determine where to
build $F_1$. Then we build all of the remaining facilities at the same
location as $F_1$. This mechanism is SGSP because it disregards the
reported aversion profile. We claim that this mechanism is
$n$-egalitarian, where $n$ denotes the number of agents. To prove this
claim, we first observe that the largest gap between two consecutive
agents with reported dislikes $\{F_1\}$ is at least
$\frac{1}{n}$. Thus the minimum welfare achieved by the mechanism is
at least $\frac{1}{2n}$.  Since the optimal minimum welfare is at most
$1/2$, the claim holds.

\subsection{The unit square}
\label{subsec:min-min-square}

In this section, we extend Mechanism~\ref{mech:min-min} to a SP
egalitarian mechanism for single-facility $\OOFLG$ when the agents are
located in the unit square.  Let $S$ denote $[0,1]^2$, let $B$ denote
the boundary of $S$, and let $x_i = (a_i, b_i)$ denote the location of
agent $i$. For convenience, we assume that all agents are located at
distinct points; the results below generalize easily to instances
where this assumption does not hold.

The analysis that we provide for our mechanism relies on results of
Toussaint~\cite{TOUSSAINT1983} concerning the largest empty circle with
location constraints problem. An instance of the latter problem is
given by a set of points in the plane.  Toussaint makes the
simplifying assumption that these points lie in general position in
the sense that no three are collinear and no four are cocircular.  In
our application of Toussaint's work, the agent locations correspond to
the set of input points. Accordingly, throughout this section, we
assume that the agent locations are in general position.\footnote{We suspect that
Toussaint's results continue to hold when the points
are not in general position. If so, we could drop
  our assumption that the agent locations are in general position.}

\begin{custommech}{11}
\label{mech:min-min-plane}
Let $I = (n, 1, \Location, \Aversion)$ denote the reported $\DOFL$
instance and let $H$ denote $\haters(I, 1)$.  If $H$ is empty, build
$F_1$ at $(0,0)$. Otherwise, construct the Voronoi diagram $D$
associated with the locations of the agents in $H$.  Let $V$ denote
the union of the following three sets of vertices: the vertices of $D$
in the interior of $S$; the points of intersection between $B$ and
$D$; the four vertices of $S$.  For each $v$ in $V$, let $d_v$ denote
the minimum distance from $v$ to any agent in $H$. Build $F_1$ at a
vertex $v$ maximizing $d_v$, breaking ties first by $x$-coordinate and
then by $y$-coordinate.
\end{custommech}

Toussaint has presented an efficient $O(n\log n)$ algorithm to find
the optimal $v$ in
Mechanism~\ref{mech:min-min-plane}~\cite{TOUSSAINT1983}. The following
lemma establishes that Mechanism~\ref{mech:min-min-plane} is
egalitarian. The lemma is shown using a result of~\cite{TOUSSAINT1983}
concerning the largest empty circle with location constraints problem.

\begin{restatable}{lemma}{minminplaneopt}
\label{lem:min-min-plane-opt}
Mechanism~\ref{mech:min-min-plane} is egalitarian for single-facility
$\OOFLG$.
\end{restatable}
\begin{proof}
Let $I = (n, 1, \Location, \Aversion)$ denote the reported $\DOFL$
instance, let $H$ denote $\haters(I, 1)$, let $y^*$ denote an optimal
location for the facility, let $\OPT$ denote $\MMW(I, y^*)$, let
$y'$ denote the location at which Mechanism~\ref{mech:min-min-plane}
builds the facility, and let $\ALG$ denote $\MMW(I, y')$.  Below we
establish that $\ALG = \OPT$, which implies that
Mechanism~\ref{mech:min-min-plane} is egalitarian.

If $H$ is empty, then it is straightforward to prove that $\ALG =
\OPT$. Otherwise, finding the optimal location at which to build
facility $F_1$ is equivalent to finding the maximum-radius circle
centered in the interior or on the boundary of $S$ such that the
interior of the circle has no points from $\{x_i \mid i \in H\}$. This
corresponds to the largest empty circle with location constraints
problem.  Toussaint shows (see Theorem~2 of~\cite{TOUSSAINT1983}) that
the optimal center for the circle is either a vertex of the Voronoi
diagram of $S$, a point of intersection of $D$ with $B$, or a vertex
of $S$.
Hence $\ALG =
\OPT$.
\end{proof}

We use a case analysis to establish
Lemma~\ref{lem:min-min-opt-sfo-plane} below.  The most interesting
case deals with an agent $i$ who dislikes $F_1$ but does not report
it.  In this case, the key insight is that when agent $i$ misreports,
facility $F_1$ is built either (1) at the same location as when agent
$i$ reports truthfully, or (2) inside or on the boundary of the
Voronoi region that contains $x_i$ when agent $i$ reports truthfully.

\begin{restatable}{lemma}{minminoptsfoplane}
\label{lem:min-min-opt-sfo-plane}
Mechanism~\ref{mech:min-min-plane} is SP for single-facility $\OOFLG$.
\end{restatable}
\begin{proof}
Assume for the sake of contradiction that
Mechanism~\ref{mech:min-min-plane} is not SP.  Thus there exists a
single-facility $\OOFLG$ instance $(I,I')$ with $I = (n, 1, \Location,
\Aversion)$ and $I'=(n, 1, \Location, \LieAversion)$, and an agent $i$
with $\LieAversion=(\Aversion_{-i}, a_i')$ who benefits by reporting
$a_i'$.  Using the same arguments as in the proof of
Lemma~\ref{lem:min-min-sp-sfo}, we conclude that $F_1$ belongs to $a_i
\setminus a'_i$.

Let $y$ (resp., $y'$) denote the location at which
Mechanism~\ref{mech:min-min-plane} builds $F_1$ when agent $i$ reports
$a_i$ (resp., $a'_i$), and let $H$ denote $\haters(I, 1)$. Note that
Mechanism~\ref{mech:min-min-plane} does not build $F_1$ at the
location of any agent in $H$, that is, $y \neq x_{i'}$ for all $i'$ in
$H$. Hence $y \neq x_i$.  When all agents report truthfully, the
Voronoi diagram partitions $S$ into $|H|$ non-overlapping polygons,
where each polygon contains one agent.  Let $P$ be the polygon that
contains agent $i$ when all agents report truthfully. When agent $i$
reports $a_i'$, the Voronoi diagram remains unchanged outside polygon
$P$.
It follows that when agent $i$ reports $a_i'$, facility $F_1$ is built
either at $y$ or at a point that belongs to $P$. If it is built at
$y$, agent $i$ does not benefit. Thus, for the remainder of the proof,
we assume that $y'$ belongs to $P$.

Let $\OPT$ (resp., $\OPT'$) denote the closest distance of any agent
in $H$ (resp., $H - i$) to the point $y$ (resp., $y'$). Let $d$ and
$d'$ denote $\dist(x_i, y)$ and $\dist(x_i, y')$, respectively. Hence
$d \ge \OPT$. Since the distance from $y$ to any agent in $H - i$ is
at least $\OPT$, we have $\OPT' \ge \OPT$ .

Since agent $i$ benefits by reporting $a_i'$, we have $d' > d$. We
begin by showing that $\OPT = \OPT'$. Suppose $\OPT \neq \OPT'$. Since
$\OPT' \ge \OPT$, we have $\OPT' > \OPT$.  Note that $\MMW(I, y') =
\min(d', \OPT')$. Since $d' > d$ and $\OPT' > \OPT$, we have $\MMW(I,
y') > \min(d, \OPT)$. Since $d \ge \OPT$, we have $\min(d, \OPT) =
\OPT$. Moreover, $\OPT = \MMW(I, y)$. Since $\MMW(I, y') > \min(d,
\OPT)$ and $\min(d, \OPT) = \MMW(I,y)$, we have $\MMW(I,y') >
\MMW(I,y)$, a contradiction since Lemma~\ref{lem:min-min-plane-opt}
implies that Mechanism~\ref{mech:min-min-plane} is egalitarian. Thus
$\OPT = \OPT'$.

Recall that $y'$ belongs to $P$. Hence the closest agent in $H$ to
$y'$ is agent $i$. Thus $d' \le \OPT'$.
Since $\OPT \le d$, $d < d'$, and $d' \le \OPT'$, we obtain $\OPT <
\OPT'$, which contradicts $\OPT = \OPT'$. Thus $d' \le d$, and hence
agent $i$ does not benefit by reporting $a_i'$.
\end{proof}

We define Mechanism~12 as the $\OOFLG$ mechanism $\Parallel{M}$, where
$M$ denotes Mechanism~\ref{mech:min-min-plane}. Using
Lemmas~\ref{lem:min-min-sfo-ooflg-sp},
\ref{lem:min-min-sfo-ooflg-opt}, \ref{lem:min-min-plane-opt},
and~\ref{lem:min-min-opt-sfo-plane}, we immediately obtain
Theorem~\ref{thm:min-min-opt-eg-plane} below.

\begin{theorem}
\label{thm:min-min-opt-eg-plane}
Mechanism~12 is SP and egalitarian.
\end{theorem}

\section{Concluding Remarks}
\label{sec:conclusion}

In this paper, we studied the obnoxious facility location game with
dichotomous preferences. This game generalizes the obnoxious facility
location game to more realistic scenarios.  All of the mechanisms
presented in this paper run in polynomial time, except that the
running time of Mechanism~\ref{mech:sum-min} has exponential
dependence on $k$ (and polynomial dependence on $n$). 
We can extend the results of Section~\ref{subsec:sum-min-square} (resp., Section~\ref{subsec:min-min-square}) to obtain an
analogue of Theorems~\ref{thm:sum-min-sgsp-square-sgsp} and~\ref{thm:sp-sum-min-sgsp-square-opt} (resp., Theorem~\ref{thm:min-min-opt-eg-plane}) that holds for an
arbitrary rectangle (resp., convex polygon).  We showed that Mechanism~\ref{mech:sum-min}
is WGSP for all $k$ and is efficient for $k \le 3$.  Properties~1
and~2 in the proof of our associated theorem,
Theorem~\ref{thm:opt-sum-min}, do not hold for $k>3$. Nevertheless, we
conjecture that Mechanism~\ref{mech:sum-min} is efficient for all $k$.
It remains an interesting open problem to reduce the gap between the
$\Omega(\sqrt{n})$ and $O(n)$ bounds for the approximation ratio
$\alpha$ of WGSP $\alpha$-egalitarian mechanisms.

\bibliographystyle{plain}
\bibliography{main}

\end{document}